\documentclass[letterpaper]{rsauthor}
\usepackage{amsfonts,amsmath,amsbsy,epsfig,natbib}
\begin{document}

\title{From Molecular Dynamics to Brownian Dynamics}

\author{\large Radek Erban}

\address{{\small Mathematical Institute, University of Oxford \\ 
Radcliffe Observatory Quarter,  Woodstock Road \\
Oxford OX2 6GG, United Kingdom \\ 
\rule{0pt}{4mm}
e-mail: erban@maths.ox.ac.uk}}

\date{\today}

\abstract{Three coarse-grained molecular dynamics (MD) models are 
investigated with the aim of developing and analyzing multiscale 
methods which use MD simulations in parts of the 
computational domain and (less detailed) Brownian dynamics (BD) simulations 
in the remainder of the domain. The first MD model is formulated 
in one spatial dimension. It is based on elastic collisions of
heavy molecules (e.g. proteins) with light point particles 
(e.g. water molecules). Two three-dimensional MD models are then
investigated. The obtained results are applied to a simplified model 
of protein binding to receptors on the cellular membrane. It is shown 
that modern BD simulators of intracellular processes can be 
used in the bulk and accurately coupled with a (more detailed) MD model 
of protein binding which is used close to the membrane.
}


\maketitle

\section{Introduction}
Brownian dynamics (BD) simulations have been used for the modelling
of a number of spatio-temporal processes in cellular and molecular 
biology in recent years, including models of intracellular 
calcium dynamics \citep{Flegg:2013:DSN}, signal trasduction 
in {\it E. coli} chemotaxis \citep{Lipkow:2005:SDP} and MAPK 
pathway \citep{Takahashi:2010:STC}. 
In these applications, trajectories and interactions between key 
biomolecules (e.g. proteins) are calculated using BD methods, while 
other components of the system (e.g. solvent molecules), which are of 
no special interest to a modeller, are not explicitly included in the 
simulation, but contribute to the dynamics of Brownian particles 
collectively as a random force. This reduces the dimensionality of 
the problem, making BD less computationally intensive than the 
corresponding molecular dynamics (MD) simulations. 

Denoting the position of a Brownian particle by 
${\mathbf X} = [X_1,X_2,X_3]$ and its diffusion constant by $D$, a simple 
model of Brownian  motion is given by the (overdamped) Langevin equation 
\begin{equation}
\mbox{d}X_i = \sqrt{2 D} \; \mbox{d}W_i, 
\qquad
i=1,2,3,
\label{BDSDE}
\end{equation}
where $W_i$, $i=1,2,3$, are three independent Wiener 
processes \citep{Erban:2007:PGS}. 
BD approaches which are based on (\ref{BDSDE}) have been implemented 
in a number of software packages designed for spatio-temporal modelling 
in systems biology, including Smoldyn \citep{Andrews:2004:SSC}, 
MCell \citep{Stiles:2001:MCM}, Green's Function Reaction
Dynamics \citep{vanZon:2005:GFR} and 
First-passage kinetic Monte Carlo method \citep{Opplestrup:2009:FKM}. 
The software package Smoldyn 
discretizes (\ref{BDSDE}) using a fixed time step $\Delta t$, i.e. it computes 
the time evolution of the position ${\mathbf X} \equiv {\mathbf X}(t)$
of each molecule by 
\begin{equation}
X_i(t+\Delta t) = X_i(t) + \sqrt{2D \Delta t} \, \xi_{i},
\qquad
i=1,2,3,
\label{eq4}
\end{equation}
where $[\xi_{1},\xi_{2},\xi_{3}]$
is a vector of normally distributed random numbers with zero mean
and unit variance. A different BD approach is implemented in
the Green's Function Reaction Dynamics \citep{Takahashi:2010:STC} 
which evolves time using a variable
time step. It approximately computes the time 
when the next reactive event happens. This means that trajectories
of molecules which are not surrounded by other reactants can 
be simulated over longer time steps.

Altough the BD models are becoming a popular choice for stochastic
modelling of intracellular spatio-temporal processes, several
difficulties prevent the use of BD for some systems. First of all,
detailed BD models are often more computationally intensive than 
coarser spatio-temporal models which are written for concentrations 
of biochemical species. In some applications (e.g. intracellular 
calcium dynamics \citep{Flegg:2013:DSN}
or actin dynamics in filopodia \citep{Erban:2013:MSR}) 
individual trajectories (computed by BD) 
are important only in certain parts of the computational domain, 
whilst in the remainder of the domain a coarser, less detailed, 
method can be used. In these applications, the computational
intensity of BD simulations can be decreased by using 
multiscale methods which efficiently and accurately
combine models with a different level of detail in 
different parts of the computational domain
\citep{Flegg:2012:TRM,Franz:2013:MRA}.

Another difficulty of BD simulations in cell and molecular biology 
is that detailed BD models require more parameters than coarser
(macroscopic) models. In some studies, 
macroscopic parameters are used to infer BD parameters 
\citep{Lipkova:2011:ABD,Andrews:2004:SSC}.
For example, knowing the macroscopic reaction rate $k$ of a bimolecular 
reaction $A + B \to C$ and diffusion constants of reactants, one can 
calculate a (microscopic) reaction radius of BD simulations which gives 
the corresponding macroscopic parameters in the limit of many particles.
In the classical Smoluchovski limit \citep{Smoluchowski:1917:VMT}, 
a bimolecular reaction occurs whenever the distance 
of reactants is less than the reaction radius 
\begin{equation}
\varrho = \frac{k}{4 \pi (D_A + D_B)}
\label{smol}
\end{equation} 
where $D_A$ (resp. $D_B$) is the diffusion constant of reactant
$A$ (resp. $B$). Although this approach is
commonly applied in stochastic reaction-diffusion models,
it is not the most satisfactory, because different microscopic
models can lead to the same macroscopic process and parameters 
\citep{Erban:2009:SMR,Lipkova:2011:ABD}. 
For example, the simplest Smoluchowski model (\ref{smol}) 
assumes that all collisions are reactive but, in reality, many non-reactive 
collisions of molecules happen before a reactive collision occurs.
Therefore, some algorithms postulate that molecules only react with 
a certain rate (probability) when the distance between reactants is less 
than a modified reaction radius (which is larger than $\varrho$).
Other methods discretize the Langevin equation with time 
step $\Delta t$ and substitute the Smoluchowski formula (\ref{smol})
(which is valid for an infinitely small time step) by a tabulated function 
computed numerically \citep{Andrews:2004:SSC}. 
However, all of these approaches are 
verified by considering the macroscopic limit (of many reactants) and 
showing that the reaction occurs with the given rate $k$
in this limit.

A different approach to parameterize BD models is to use a more 
detailed description written in terms of MD. In this paper, we  
investigate connections between BD and MD models with the aim of 
developing and analyzing of multiscale methods which couple BD 
and MD simulations. 
We consider a (computationally intensive) MD simulation in 
domain $\Omega$ which is either one-dimensional or
three-dimensional, i.e. $\Omega \subset {\mathbb R}$ or
$\Omega \subset {\mathbb R}^3$. Our main goal is to design and analyse 
multiscale methods which can compute spatio-temporal statistics 
with MD-level of detail in the subdomain $\Omega_D \subset \Omega$. 
We define 
\begin{equation}
\Omega_C = \Omega \setminus \overline{\Omega_D},
\qquad
I = \partial \Omega_D \cap \partial \Omega_C, 
\label{geom1}
\end{equation} 
where $\Omega_D$ and $\Omega$ are open sets, 
the (open) set $\Omega_C$ is the complement of $\Omega_D$ 
and $I$ is the shared interface (boundary) between 
$\Omega_D$ and $\Omega_C$. In the multiscale set up
(\ref{geom1}), we use a detailed MD model in
$\Omega_D$ and a coarser BD model in $\Omega_C$. 

In this paper, we focus on a simple MD approach which is 
introduced in Sections \ref{secMDmodelA} and \ref{secMDmodelBC}.
A few (heavy) particles with mass $M$ and radius $R$ 
are coupled with a large number of light point particles with masses 
$m \ll M$. The collisions of particles are without friction, which 
means that post-collision velocities can be computed using the
conservation of momentum and energy \citep{Holley:1971:MHP,Durr:1981:MMB}.
We will introduce and study three MD models which make use
of elastic collisions. They will be denoted as MD models 
[A], [B] and [C] in what follows.
More complicated MD approaches are discussed 
in Section \ref{secdiscussion}.

The first MD model [A] is introduced in Section \ref{secMDmodelA}. 
It is a one-dimensional MD model where all particles move 
along the real line. In particular, the radius $R$ do not 
have to be considered, because it has no influence on the 
dynamics of large particles. In one dimension, heat 
bath particles cannot pass each other, which makes the MD model [A] 
different from three-dimensional models in Section \ref{secMDmodelBC}
where heat bath particles (points) do not interact with each other. 

In Section \ref{secMDmodelBC}, we 
introduce two three-dimensional models, denoted [B] and [C], 
where the nonzero radius $R$ is one of  the key parameters. 
To make one-dimensional and three-dimensional models comparable, 
we keep $R$ fixed in the three-dimensional model and we study
the behaviour of all MD models in the limit $M/m \to \infty.$ This limit
can be achieved in many different ways. For example, we can keep 
$m$ fixed and pass $M \to \infty$, or we can keep $M$ fixed and 
pass $m \to 0$. In what follows we define the parameter
\begin{equation}
\mu = \frac{M}{m}.
\label{defmu}
\end{equation}
This parameter is dimensionless, even if we assume that $M$ and $m$
have physical units of mass. However, in this paper, all parameters
are considered dimensionless for simplicity. We are interested 
in the limit $\mu \to \infty$.

All three models [A], [B] and [C] converge in apropriate
limits to the Brownian motion of large particles given by 
(\ref{BDSDE}). One can also show that these models converge to
the Langevin description \citep{Holley:1971:MHP,Durr:1981:MMB,Dunkel:2006:RBM}
\begin{eqnarray}
\mbox{d}X_i & = & V_i \; \mbox{d}t, 
\label{BDXeq}
\\
\mbox{d}V_i & = & - \gamma \, V_i \, \mbox{d}t 
+ \gamma \sqrt{2 D} \; \mbox{d}W_i, 
\quad
i=1,2,3,
\label{BDVeq}
\end{eqnarray}
where $[X_1,X_2,X_3]$ is the position of a diffusing
molecule, $[V_1,V_2,V_3]$ is its velocity, $D$ is
the diffusion coefficient and $\gamma$ is the friction coefficient. 
This description can be further reduced to 
(\ref{BDSDE}) in the overdamped limit $\gamma \to \infty$. 
We overview the results 
which relate MD models [A], [B] and [C] with Brownian motion 
in Sections \ref{secMDmodelA} and \ref{secMDmodelBC}.

Both (\ref{BDSDE}) or (\ref{BDXeq})--(\ref{BDVeq}) reduce the 
dimensionality of the problem, making BD less computationally intensive 
than the corresponding MD simulations. In Sections \ref{secfromMDtoBD1D}
and \ref{secfromMDtoBD3D}, we study how MD models [A], [B] and [C]
can be used in one part $\Omega_D$ of the computational domain $\Omega$
and the BD models (\ref{BDSDE}) or (\ref{BDXeq})--(\ref{BDVeq})
in the remainder $\Omega_C$, making use of the notation (\ref{geom1}). 
We apply our findings to a simplified model of protein binding to
receptors in Section \ref{bioapplication}. We conclude with discussing
our results in Section \ref{secdiscussion}.

\section{One-dimensional MD model [A]}
\label{secMDmodelA}

\noindent
The MD model [A] is described in terms of positions $x^i$ 
and velocities $v^i$, $i=1,2,3,\dots$, of heat bath particles,
and positions $X^i$ and velocities $V^i$, $i=1,2,\dots,N$, 
of heavy particles of mass $M \gg m$, where $m$ is the mass
of a heat bath particle. In our computer implementations,
we will consider a finite number of heat bath particles.
However, we formulate the MD model in terms of (countably)
infinitely many of heat bath particles which are initially
distributed along the real line according to the Poisson 
distribution with density
\begin{equation}
\lambda_\mu = \frac{1}{4} \sqrt{\frac{\pi(\mu+1)\gamma}{2 D}},
\label{distr1Dx}
\end{equation}
where $\mu$ is given by (\ref{defmu}), and $D$ and $\gamma$ are
positive constants. This means that the probability that
there are $j$ particles in a subinterval $[a,b] \subset {\mathbb R}$, 
$a < b$, is equal to 
$$
\frac{(\lambda_\mu (b-a))^j}{j!} \exp \big[ - \lambda_\mu (b-a) \big],
$$
where $(b-a)$ is the lentgth of the interval $[a,b]$.
Initial velocities of heat bath particles 
are given by the normal distribution
\begin{equation}
f_\mu(v) = \frac{1}{\sigma_\mu \sqrt{2 \pi}}
\exp \left( - \frac{v^2}{2 \sigma_\mu^2} \right),
\qquad \mbox{where} \qquad
\sigma_\mu = \sqrt{(\mu+1) D \gamma}.
\label{distr1Dv}
\end{equation}
Let us consider a model with a single heavy particle, 
i.e. $N=1$. Then its location $X^1$ and velocity $V^1$ will
be denoted as $X$ and $V$ to simplify our notation. Whenever
the heavy particle collides with the light particle 
with velocity $v^i$, their velocities are updated 
using the conservation of mass and momentum:
\begin{eqnarray}
{\widetilde{V}}
&=& 
\frac{M - m}{M + m} \, V
 + 
\frac{2 m}{M + m} \, v^i,
\label{col1DV} \\
{\widetilde{v}}^{i} &=& \frac{m - M}{M + m} \, v^i
 + 
\frac{2 M}{M + m} \, V,
\label{col1Dv}
\end{eqnarray}
where tildes denote post-collision velocities. Using (\ref{defmu}),
the equations (\ref{col1DV})--(\ref{col1Dv}) can be rewritten as
\begin{equation}
{\widetilde{V}}
= 
\frac{\mu - 1}{\mu + 1} \, V
 + 
\frac{2}{\mu + 1} \, v^i, \qquad\qquad
{\widetilde{v}}^{i}
= \frac{1 - \mu}{\mu + 1} \, v^i
 + 
\frac{2 \mu}{\mu + 1} \, V.
\label{col1Dmu}
\end{equation}
The following result can be shown for the above MD model [A]:

\begin{lemma} \label{convergenceelastic}
Let $\gamma > 0$ and $D > 0$. Let us consider the heavy particle
of mass $M$ with initial position $X_\mu(0) = X_0$ and initial velocity
$V_\mu(0) = V_0$ which is subject to elastic collisions 
$(\ref{col1Dmu})$ with heat bath particles
of mass $m$ whose initial positions and velocities are distributed
according to $(\ref{distr1Dx})$--$(\ref{distr1Dv})$. Then the
$X_\mu$ and $V_\mu$ converges (as $\mu \to \infty$) in distribution 
to the solution $X$ and $V$ of equations 
\begin{equation}
\mbox{{\rm d}}X = V \, \mbox{{\rm d}}t
\qquad
\mbox{and} 
\qquad
\mbox{{\rm d}}V = - \gamma \, V + \gamma \sqrt{2 D} \; 
\mbox{{\rm d}}W,
\label{evollimitXV}
\end{equation} 
where $X(0)=X_0$ and $V(0)=V_0$.
That is, $X$ and $V$ solve the one-dimensional version  
of equations $(\ref{BDXeq})$--$(\ref{BDVeq})$.
\end{lemma} 

\begin{proof}
This lemma can be proven using the main theorem in \citet{Holley:1971:MHP}
where it is shown that a similar process converges to the Ornstein-Uhlenbeck
process $(\ref{BDVeq})$ for velocities. Although our funcion $f_\mu$
does not satisfy all assumptions of the main theorem of 
\citet{Holley:1971:MHP}, a simple rescaling of our parameters 
leads to a process which is covered by Holley's theorem.
In Section \ref{secfromMDtoBD1D} of this paper, we also
rederive this result as one of the consequences of 
multiscale analysis, see (\ref{fullSDE3}).
\end{proof}

Since the goal of this paper is to study the behaviour of computational
algorithms, we formulate the MD model [A] in a finite domain
$[-L,L]$, i.e. we consider a finite number 
$n \equiv n(t)$ of heat bath particles which are at positions
$x^i \in [-L,L]$ with velocities $v^i \in (-\infty,\infty)$, 
$i = 1,2,\dots,n$. We want to formulate boundary
conditions of our problem so that the spatio-temporal statistics 
in $[-L,L]$ are equivalent to spatio-temporal statistics of the
original unbounded process. The following lemma will be useful for 
designing appropriate boundary conditions.

\begin{lemma} \label{lemmaboundary}
Let $b \in {\mathbb R}$ and $\Delta t > 0$. 
Let us assume that heat bath particles are distributed 
according to the Poisson distribution with density $(\ref{distr1Dx})$ in 
the interval $(-\infty,b)$. Their initial velocities are given
according to $(\ref{distr1Dv})$ and there are no particles in the
interval $(b,\infty)$ at time $t=0$. Then the average number 
of particles in the interval $(b,\infty)$ at time $t=\Delta t$ is 
\begin{equation}
\frac{\gamma(\mu+1)\Delta t}{8}. 
\label{averpartdeltat}
\end{equation}
The positions $x$ and velocities $v$ of these particles are distributed 
according to
\begin{equation}
H(b - x + v \Delta t) \, \lambda_\mu \, f_\mu(v),
\label{distveldeltat}
\end{equation}
where $H(\cdot)$ is the Heaviside step function.
In particular, the positions of the particles at point 
$x \in (b,\infty)$ are distributed at time $t = \Delta t$
according to 
\begin{equation}
\varrho(x;\Delta t,b) 
\equiv
\frac{\lambda_\mu}{2} \, \mathrm{erfc} \left(
\frac{x-b}{\Delta t \, \sigma_\mu \sqrt{2}} \right),
\qquad
\mbox{for}
\;\;
x \in (b,\infty),
\label{distpartdeltat}
\end{equation}
where 
$\mathrm{erfc}(z) = 2/\sqrt{\pi} \int_z^\infty \exp(-s^2) \, \mbox{{\rm d}}s$ 
is the complementary error function.
\end{lemma}

\begin{proof}
Particles which are at point $x \in (b,\infty)$ at time $t = \Delta t$
were previously at point $x-v \Delta t$ at time $t=0$. In particular,
there will be nonzero heat bath particles with velocity $v$ at point 
$x$ at time $t = \Delta t$ provided that $x-v \Delta t < b$
which implies (\ref{distveldeltat}). Consequently,
the density of particles which are at point $x \in (b,\infty)$
at time $t = \Delta t$ is
\begin{eqnarray*}
\varrho(x;\Delta t,b)
&=&
\int_{-\infty}^\infty
H(b - x + v \Delta t) \, \lambda_\mu \, f_\mu(v)
\, \mbox{d} v
=
\int_{(x-b)/{\Delta t}}^\infty \lambda_\mu f_\mu(v)
\, \mbox{d} v
\\
&=&
\frac{\lambda_{\mu}}{\sqrt{2 \pi \sigma_\mu^2}}
\int_{(x-b)/{\Delta t}}^\infty
\exp \left( - \frac{v^2}{2 \sigma_\mu^2} \right) \, \mbox{d} v
=
\frac{\lambda_\mu}{2} \, \mbox{erfc} \left(
\frac{x-b}{\Delta t \, \sigma_\mu \sqrt{2}} \right).
\end{eqnarray*} 
Thus we proved (\ref{distpartdeltat}). Integrating this formula
over $x$ in interval $(b,\infty)$, we obtain the average 
number of particles which are in the interval $(b,\infty)$
at time $t = \Delta t$:
$$
\int_b^{\infty} 
\varrho(x;\Delta t,b)
\, \mbox{d} x
=
\frac{\lambda_\mu}{2}
\int_0^{\infty} 
\mbox{erfc} \left(
\frac{z}{\Delta t \, \sigma_\mu \sqrt{2}} \right)
\mbox{d} z
=
\frac{\lambda_\mu \sigma_\mu \Delta t}{\sqrt{2\pi}}.
$$ 
Substituting (\ref{distr1Dx}) for $\lambda_\mu$ and 
(\ref{distr1Dv}) for $\sigma_\mu$, we obtain
(\ref{averpartdeltat}).
\end{proof}

We use Lemma \ref{convergenceelastic} and
Lemma \ref{lemmaboundary} to design a computational 
test for multiscale methods. Since the number $n \equiv n(t)$ 
of heat bath particles in $[-L,L]$ is much larger than the 
number $N$ of large particles, we will focus on models of 
a single large particle, i.e. $N=1$, which is described by its
position $X$ and velocity $V$. We choose a small
time step $\Delta t$. One iteration of the MD algorithm is 
presented in Table \ref{tablealgM1}.
\begin{table}
\framebox{%
\hsize=0.97\hsize
\vbox{
\leftskip 10mm
\parindent -10mm
[A1] \hskip 1.2mm
Compute ``free-flight positions'' of heat bath particles 
and the large particle at time $t+\Delta t$ by: 

\smallskip

\noindent
${\widehat{x}}^{i}(t+\Delta t) = x^i(t) + v^i(t) \, \Delta t$
and ${\widehat{X}}(t+\Delta t) = X(t) + V(t) \, \Delta t.$ 

\smallskip

[A2] \hskip 1.2mm
Compute post-collision velocities by (\ref{col1Dmu})
for every pair of particles which collided. Compute their 
post-collision positions $x^{i}(t+\Delta t)$ 
and $X(+\Delta t)$ by updating their 
``free-flight positions'' ${\widehat{x}}^{i}(t+\Delta t)$ and
${\widehat{X}}(t+\Delta t).$

\smallskip

[A3] \hskip 1.2mm
Terminate trajectories of heat bath particles which 
left the domain $[-L,L]$. Update $n$ accordingly. 

\smallskip

[A4] \hskip 1.2mm
Generate a random number $r_1$ uniformly distributed
in $(0,1)$. \\
If $r_1 < \gamma (\mu + 1) \Delta t/8$, then
increase $n$ by 1, and introduce a new heat bath particle at 
a position sampled according to the probability distribution 
proportional to $\varrho(x;\Delta t,-L)$. Its velocity 
is sampled according to the probability distribution proportional to 
$H(-L - x + v \Delta t) \, f_\mu(v)$.

\smallskip

[A5] \hskip 1.2mm
Generate a random number $r_2$ uniformly distributed
in $(0,1)$. \\
If $r_2<\gamma (\mu + 1) \Delta t/8$, then
increase $n$ by 1, and introduce a new heat bath particle at position
$x^n(t+\Delta t)$ with velocity $v^n(t+\Delta t)$ which are 
sampled according to probability distributions 
(\ref{samplingnewXright}) and (\ref{samplingnewVright}).

\smallskip

[A6] \hskip 1.2mm
Continue with the step [A1] using time $t=t+\Delta t$.

\par \vskip 0.8mm}
}\vskip 1mm
\caption{One iteration of the computer implementation of MD model [A].}
\label{tablealgM1}
\end{table} 
We first compute the positions of all particles at time $t+\Delta t$
in the step [A1] by assuming that particles do 
not interact. Then we use (\ref{col1Dmu}) to incorporate collisions 
in the step [A2]. Since all heat bath
particles have the same mass, the collisions between them result
in exchange of colliding particles' positions and velocities. In particular
the step [A2] can be implemented by sorting 
the heat bath particles during every iteration. All particles which 
left the domain $[-L,L]$ are removed in the step 
[A3]. 

New heat bath particles are introduced in the steps [A4]
and [A5]. We assume that $\Delta t$ 
is chosen so small that (\ref{averpartdeltat}) is much smaller than 1.
Then (\ref{averpartdeltat}) can be interpretted as a probability of 
introducing one particle from the left (resp. right) during one timestep. Using 
Lemma \ref{lemmaboundary}, the new particle will be introduced at 
the (left boundary) position which is sampled according to the probability 
distribution proportional to $\varrho(x;\Delta t,-L)$ in the step [A4]. 
To sample from this probability distribution, 
we scale and shift a random number sampled from the 
complementary error function distribution $\pi \, \mathrm{erfc}(z)$ 
where $z \in (0,\infty)$. An acceptance-rejection algorithm for sampling 
random numbers from $\pi \, \mathrm{erfc}(z)$ is given in Table
\ref{tableerfcsampling}. We use it with the constants $a_1$ and
$a_2$ given by
\begin{equation}
a_1 = 0.532, 
\qquad
\mbox{and}
\qquad
a_2 = 0.814.
\label{valuesa1a2}
\end{equation} 
The values of constants $a_1$ and $a_2$ were computed to maximize the total
acceptance probability of the acceptance-rejection algorithm in 
Table \ref{tableerfcsampling}.
Using (\ref{valuesa1a2}), we accept 86\% of proposed numbers $\zeta_2$.

\begin{table}
\framebox{%
\hsize=0.97\hsize
\vbox{
\leftskip 4mm
\parindent -4mm

$\bullet\;$
Generate a random number $\zeta_1$ uniformly distributed in (0,1).

\smallskip

$\bullet\;$
Compute exponentially distributed random number $\zeta_2$ by
$\zeta_2 = - a_1 \, \log(\zeta_1).$

\smallskip

$\bullet\;$
Generate a random number $\zeta_3$ uniformly distributed in (0,1).

\smallskip

$\bullet\;$
If $\zeta_1 \, \zeta_3 < a_2 \, \mathrm{erfc} (\zeta_2)$,
then choose $\zeta_2$ as a sample from the probability distribution
$\pi \, \mathrm{erfc}(z)$. Otherwise, repeat the algorithm.

\par \vskip 0.8mm}
}\vskip 1mm
\caption{The acceptance-rejection algorithm which is used to sample
random numbers which are distributed according to the probability 
distribution $\pi \, \mathrm{erfc}(z)$ where
$z \in (0,\infty)$. In our simulations, we use constants
$a_1$ and $a_2$ given by (\ref{valuesa1a2}).}
\label{tableerfcsampling}
\end{table} 

A particle introduced close to the right boundary in the step 
[A5] will have its position sampled according 
to the probability distribution
\begin{equation}
C_1 \, \mathrm{erfc} \left(\frac{L-x}{\Delta t \, \sigma_\mu \sqrt{2}} \right),
\qquad
\mbox{for}
\;\;
x \in (-\infty,L),
\label{samplingnewXright}
\end{equation}
where $C_1$ is a normalization constant. The probability 
distribution (\ref{samplingnewXright}) is proportional to 
$\varrho(-x;\Delta t,-L)$ and can be justified using the
same argument as Lemma \ref{lemmaboundary}. 
To sample from the probability distribution (\ref{samplingnewXright}), 
we again use the acceptance-rejection algorithm in 
Table \ref{tableerfcsampling} with parameters $a_1$ and
$a_2$ given by (\ref{valuesa1a2}). In the step 
[A5], we also sample the velocity 
$v \in {\mathbb R}$ of the new particle using the truncated 
Gaussian distribution
\begin{equation}
C_2 \, H(x - v \Delta t - L) \, f_\mu(v),
\label{samplingnewVright}
\end{equation}
where $C_2$ is a normalization constant.
To sample random numbers according to the truncated normal 
distributions in the steps [A4] and [A5], we use an
acceptance-rejection algorithm which is derived as
Proposition 2.3 in \cite{Robert:1995:STN}.

\newcommand{\picturesAB}[4]{
\centerline{
\hskip #4
\raise #3 \hbox{\raise 0.9mm \hbox{(a)}}
\hskip -8mm
\epsfig{file=#1,height=#3}
\raise #3 \hbox{\raise 0.9mm \hbox{(b)}}
\hskip -8mm
\epsfig{file=#2,height=#3}
}}

In Figure \ref{figure1}, we present illustrative results computed
by the algorithm [A1]--[A6].
We use $\mu = 10^3,$ $\gamma = 10$ and $D=1$. 
We initialize the position and velocity of the heavy particle
as $X(0) = 0$ and $V(0) = 0$ and we use the algorithm
[A1]--[A6] with time step $\Delta t = 10^{-7}$
in the interval $[-L,L]$ where $L=20$. 
\begin{figure}[t]
\picturesAB{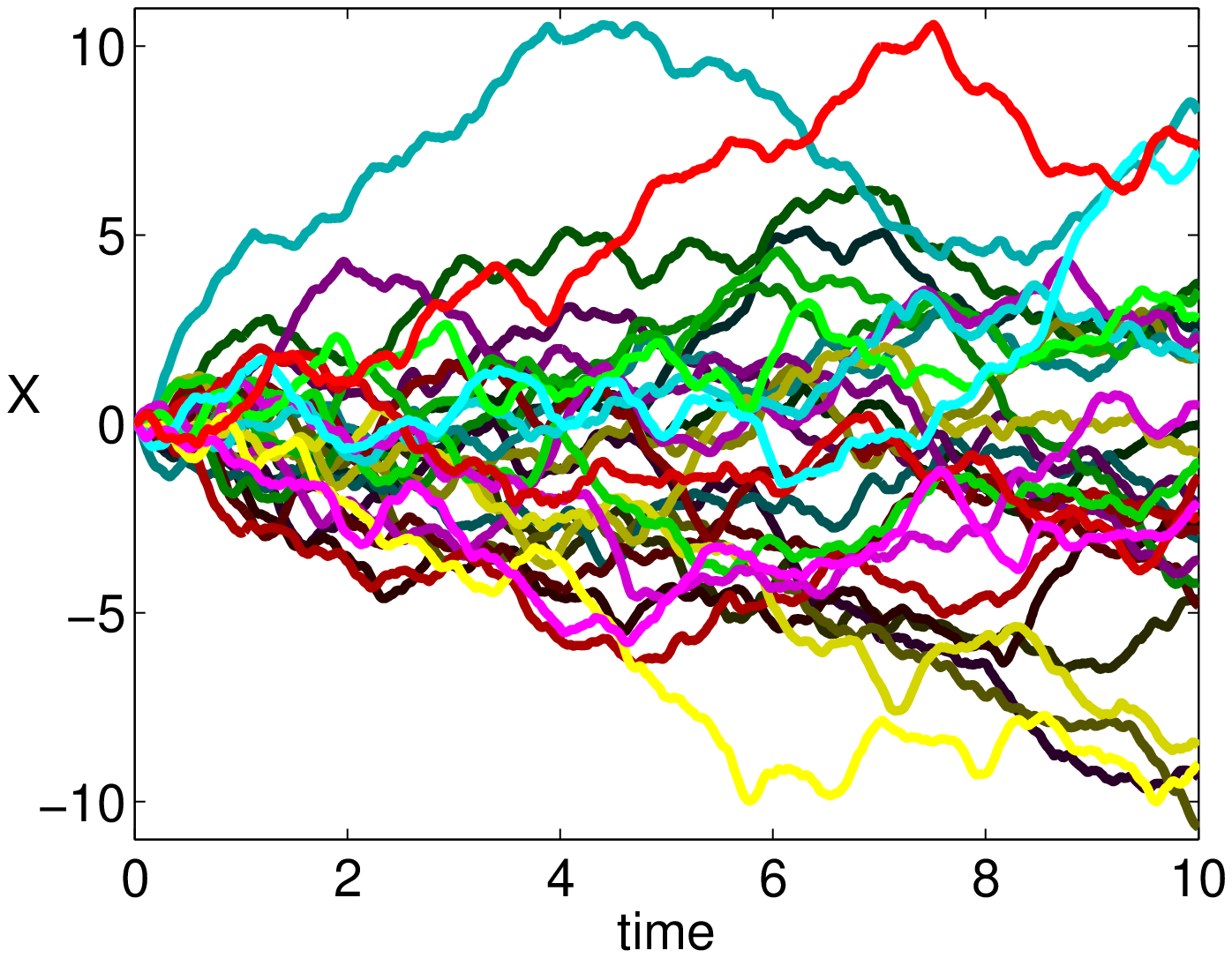}{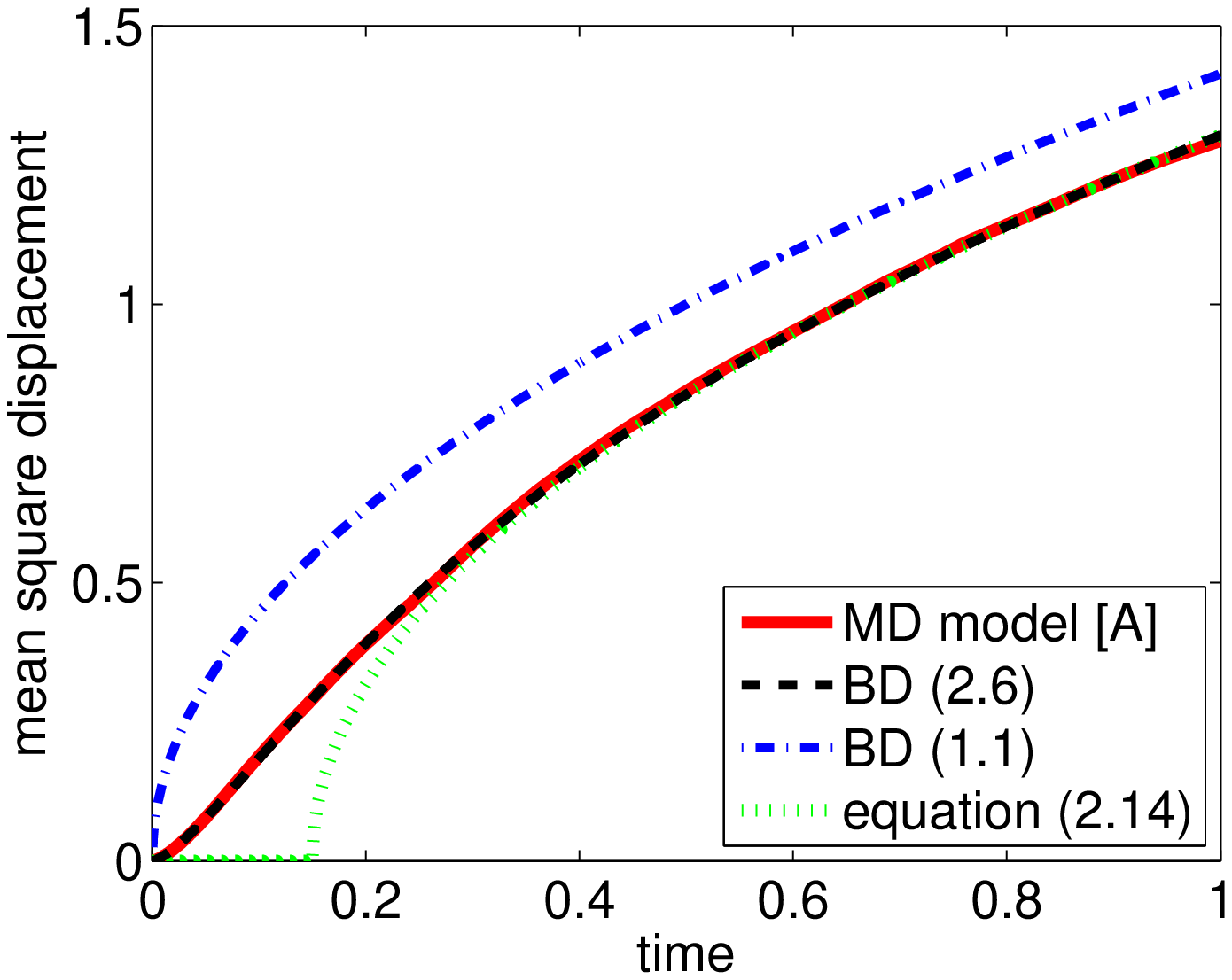}{5.3cm}{3mm}
\caption{(a) {\it Thirty illustrative trajectories of the heavy
particle computed by the MD algorithm {\rm [A1]}--{\rm [A6]}.}
(b) {\it The mean square displacement computed by
$10^3$ realizations of the algorithm {\rm [A1]}--{\rm [A6]}
(red solid line). The MD results are compared with BD results:
equation $(\ref{eqmsd15})$ (black dashed line), 
$\sqrt{2 D \, t}$ (blue dot-dashed line)
and equation $(\ref{eq2p14})$ (green dotted line). \hfill\break
We use $\mu = 10^3,$ $\gamma = 10$, $D=1$, $\Delta t = 10^{-7}$,
$L=20$, $X(0) = 0$ and $V(0) = 0$.}
}
\label{figure1}
\end{figure}
In Figure \ref{figure1}(a), we present 30 illustrative trajectories
of the heavy particle $X(t)$ computed for
$t \in [0,10]$. The mean square displacement given by the MD
model [A] is plotted in Figure \ref{figure1}(b) as the red solid
line. To illustrate the limiting result in Lemma \ref{convergenceelastic}, 
we also plot the mean square displacement corresponding to the limiting
solution $X$ of (\ref{evollimitXV}). It can be analytically computed as
\begin{equation}
\sqrt{{\mathbb E} 
\left[ 
\left( X(t)-X(0) \right)^2
\right]}
=
\sqrt{2 D \, t 
-
\frac{3 D}{\gamma}
+ 
\frac{4 D \, \exp [- \gamma t]}{\gamma} 
-
\frac{D \, \exp [- 2 \gamma t]}{\gamma}
},
\label{eqmsd15}
\end{equation}
where ${\mathbb E}[\cdot]$ denotes the expected value.
It is plotted as the black dashed line in Figure \ref{figure1}(b).
We also plot the mean square displacement corresponding to 
the overdamped limit (\ref{BDSDE}), i.e. $\sqrt{2 D \, t}$,
as the blue dot-dashed line in Figure \ref{figure1}(b). If we neglect
the exponential terms in (\ref{eqmsd15}), we obtain
\begin{equation}
\sqrt{{\mathbb E} 
\left[ 
\left( X(t)-X(0) \right)^2
\right]}
\approx
\sqrt{2 D \, 
\left( t 
-
\frac{3}{2 \gamma}
\right)}.
\label{eq2p14}
\end{equation}
This approximation is plotted in Figure \ref{figure1}(b) as the green
dotted line. We will use (\ref{eq2p14}) later in Section \ref{bioapplication}
to couple the overdamped BD model (\ref{BDSDE}) with MD simulations.

\section{Three-dimensional MD models [B] and [C]}
\label{secMDmodelBC}

MD models [B] and [C] are three-dimensional generalizations of the MD 
model [A]. They are described in terms of positions ${\mathbf x}^i$ 
and velocities ${\mathbf v}^i$, $i=1,2,3,\dots$, of heat bath particles,
and positions ${\mathbf X}^i_\mu = [X_{\mu;1}^i,X_{\mu;2}^i,X_{\mu;3}^i]$ 
and velocities ${\mathbf V}^i_\mu = [V_{\mu;1}^i,V_{\mu;2}^i,V_{\mu;3}^i]$, 
$i=1,2,\dots,N$, of heavy particles of mass $M \gg m$, where $m$ is the mass
of a heat bath particle. We again define $\mu$ by (\ref{defmu}).
We will denote by $R$ the radius of a heavy particle.

MD models [B] and [C] are both based on elastic collisions of heavy 
molecules (balls with mass $M$ and radius $R$) with point
bath particles with masses $m$. Since the collisions are without
friction, conservation of momentum and energy then yields the
following generalization of formulae (\ref{col1Dmu}) for
post-collision velocities \citep{Durr:1981:MMB}
\begin{eqnarray}
\left[\widetilde{\mathbf V}^{i}_\mu \right]^\prime 
&=& 
\left[ {\mathbf V}^{i}_\mu \right]^\parallel
+ 
\frac{\mu - 1}{\mu + 1} \, \left[ {\mathbf V}^{i}_\mu \right]^\perp
 + 
\frac{2}{\mu + 1} \, \left[{\mathbf v}^j \right]^\perp,
\label{eq11}
\\
\left[\widetilde{\mathbf v}^j\right]^\prime 
&=& 
\left[ {\mathbf v}^j \right]^\parallel
+
\frac{1 - \mu}{\mu + 1} \, \left[{\mathbf v}^j \right]^\perp
 + 
\frac{2 \mu}{\mu + 1} \, \left[ {\mathbf V}^{i}_\mu \right]^\perp,
\label{eq12}
\end{eqnarray}
where ${\mathbf v}^j$ is the velocity of the heat bath molecule
which collided with the $i$-th heavy molecule, tildes denote 
post-collision velocities, superscripts $\perp$ denote projections of 
velocities on the line through the centre of the molecule and the 
collision point on its surface, and superscripts $\parallel$ denote
tangential components. 

\subsection{MD model {\rm [B]}}
\label{subsecMDmodelB}

\noindent
MD model [B] will use the normal distribution for velocities of heat
bath particles. The following lemma generalizes 
Lemma \ref{convergenceelastic} to the three-dimensional MD model [B].

\begin{lemma}\label{3Dlemmaexp}
Let $\gamma>0$, $D>0$ and $R>0$. Let us consider the MD model {\rm [B]} 
where heat bath particles are distributed according
to the Poisson distribution with density
\begin{equation}
\lambda_\mu =
\frac{3}{8 R^2} \sqrt{\frac{(\mu + 1) \gamma}{2 \pi D}}.
\label{lambda3Dexp}
\end{equation}
Let the velocities of heat bath particles are distributed according
to
\begin{equation}
f_\mu({\mathbf v})
= \frac{1}{\sigma_\mu^3(2\pi)^{3/2}} 
\exp \left[ 
- \frac{v_1^2+v_2^2+v_3^2}{2 \sigma_\mu^2} \right],
\quad \mbox{where} \quad
\sigma_\mu = \sqrt{(\mu + 1) \, D \, \gamma }
\label{fvel3Dexp}
\end{equation}
and ${\mathbf v} = [v_1,v_2,v_3].$ We will consider one heavy molecule 
in such a heat bath, i.e. $N=1$. Then the position and
velocity of the heavy molecule, ${\mathbf X}_\mu$ and ${\mathbf V}_\mu$,
converge (in the sense of distribution) to the solution of
$(\ref{BDXeq})$--$(\ref{BDVeq})$ in the limit $\mu \to \infty.$ 
\end{lemma}

\begin{proof}
\noindent
The MD model [B] and heat bath distributions (\ref{lambda3Dexp})
and (\ref{fvel3Dexp}) satisfy the assumptions of Theorem 2.1 in 
\cite{Durr:1981:MMB}. Their theorem expresses the limiting 
equation of a process with given $\lambda_\mu$ and 
$f_\mu({\mathbf v})$ in terms of moments of $f_\mu$.
These moments can be analytically evaluated to verify the 
statement of Lemma \ref{3Dlemmaexp}. We will also rederive
this result in Section \ref{secfromMDtoBD3D} as a consequence
of the analysis of multiscale methods.
\end{proof}

Lemma \ref{3Dlemmaexp} can be viewed as a different formulation 
of Theorem 2.1 in \cite{Durr:1981:MMB}. They were interested
in the limit $m \to 0$ which is equivalent to $\mu \to \infty$.
Considering the scaling $m^{3/2} f({\mathbf v} m^{1/2})$
of the velocity distribution of heat bath particles
(with density scaled as $\lambda/m^{1/2}$), they derived 
formulae for $\gamma$ and $D$ in terms of moments of $f$
and $\lambda$. To formulate Lemma \ref{3Dlemmaexp}, we
inverted their results by deriving the appropriate distributions 
(\ref{lambda3Dexp}) and (\ref{fvel3Dexp}) which lead to the
limiting BD model with a given $D$ and $\gamma$. 

\subsection{MD model  {\rm [C]}}
\label{subsecMDmodelC}

\noindent
In Lemma \ref{3Dlemmaexp} we used the normal distribution
for velocities (\ref{fvel3Dexp}). Another option is to use
heat bath particles with fixed speed as it is done in 
the following Lemma \ref{3Dlemmafixedspeed}. We denote the 
resulting MD model as the MD model [C].
 
\begin{lemma}\label{3Dlemmafixedspeed}
Let $\gamma>0$, $D>0$ and $R>0$. Let us consider the MD model {\rm [C]} 
where heat bath particles are distributed according
to the Poisson distribution with density
\begin{equation}
\lambda_\mu =
\frac{3}{8 \pi R^2} \sqrt{\frac{(\mu + 1) \gamma}{D}}.
\label{lambda3Dfixedspeed}
\end{equation}
Let the velocities of heat bath particles are distributed according
to
\begin{equation}
f_\mu({\mathbf v})
= \frac{1}{4 \pi \sigma_\mu^2} 
\, \delta \left( \sqrt{v_1^2+v_2^2+v_3^2} - \sigma_\mu \right),
\quad \mbox{where} \quad
\sigma_\mu = 2 \sqrt{(\mu + 1) \, D \, \gamma}
\label{fvel3Dfixedspeed}
\end{equation}
and $\delta$ is a Dirac distribution.
Let us consider one heavy molecule in this heat bath at position 
${\mathbf X}_\mu$ with velocity ${\mathbf V}_\mu$. 
Then ${\mathbf X}_\mu$ and ${\mathbf V}_\mu$ converge (in the sense 
of distribution) to the solution of $(\ref{BDXeq})$--$(\ref{BDVeq})$ 
in the limit $\mu \to \infty.$ 
\end{lemma}

Lemma \ref{3Dlemmafixedspeed} can again be proven using Theorem 2.1 
in \cite{Durr:1981:MMB} which is applicable to any spherically 
symmetric velocity distribution which has at least four finite moments.

\subsection{Boundary conditions for MD models {\rm [B]} and {\rm [C]}}
\label{subsecMDmodelBCboundary}

\noindent
Next, we generalize Lemma \ref{lemmaboundary} to the three-dimensional
case. This will help us to specify boundary conditions for simulations 
which use the MD models [B] and [C] in finite domains. 

\begin{lemma} \label{lemmaboundaryBC}
Let $b \in {\mathbb R}$ and $\Delta t > 0$. 
Let us assume that heat bath particles are distributed 
according to the Poisson distribution with density 
$\lambda_\mu$ in the half space $(-\infty,b) \times {\mathbb R}^2$; 
their initial velocities are distributed according to 
$f_\mu({\mathbf v})$ and there are no particles in 
the half space $(b,\infty) \times {\mathbb R}^2$ 
at time $t=0$. Let us assume that $\lambda_\mu$ and $f_\mu({\mathbf v})$ are
either given by $(\ref{lambda3Dexp})$--$(\ref{fvel3Dexp})$
(MD model {\rm [B]}), 
or by $(\ref{lambda3Dfixedspeed})$--$(\ref{fvel3Dfixedspeed})$
(MD model {\rm [C]}).

Then the positions ${\mathbf x}$ and velocities 
${\mathbf v}$ of heat bath particles in the half space 
$(b,\infty)  \times {\mathbb R}^2$ are distributed at time 
$t = \Delta t$ according to 
\begin{equation}
H(b - x_1 + v_1 \Delta t) \, \lambda_\mu \, f_\mu({\mathbf v}),
\label{distveldeltatBC}
\end{equation}
and the average number of particles in the semi-infinite cuboid 
$(b,\infty)  \times (0,1)^2$ at time $t=\Delta t$ is 
\begin{equation}
\frac{3 \gamma (\mu + 1) \Delta t}{16 \pi R^2}.
\label{averpartdeltatBC}
\end{equation}
\end{lemma}

\begin{proof}
Formula (\ref{distveldeltatBC}) is a generalization of formula 
(\ref{distveldeltat}) in Lemma \ref{lemmaboundary} and can
be justified using the same arguments. To prove (\ref{averpartdeltatBC}),
we will distinguish two cases. 

First, let us consider that 
$\lambda_\mu$ and $f_\mu({\mathbf v})$ are given by 
$(\ref{lambda3Dexp})$--$(\ref{fvel3Dexp})$.
Integrating (\ref{distveldeltatBC}) over positions and
velocities (see the proof of Lemma \ref{lemmaboundary}), 
we conclude that the average number of particles
in the semi-infinite cuboid $(b,\infty)  \times (0,1)^2$
at time $t=\Delta t$ is in the case (a) equal to
$$
\frac{\lambda_\mu \sigma_\mu \Delta t}{\sqrt{2\pi}}
=
\frac{3 \gamma (\mu + 1) \Delta t}{16 \pi R^2}
\label{averpartdeltatBCal}
$$
which is the formula (\ref{averpartdeltatBC}). 

Next, let us consider that 
$\lambda_\mu$ and $f_\mu({\mathbf v})$ are given by
$(\ref{lambda3Dfixedspeed})$--$(\ref{fvel3Dfixedspeed})$.
Integrating (\ref{distveldeltatBC}) with respect of
${\mathbf v}$, we get the density of particles
at ${\mathbf x} \in (b,\infty)  \times {\mathbb R}^2$ 
at time $t=\Delta t$:
\begin{eqnarray}
\varrho({\mathbf x};\Delta t,b)
&=&
\int_{{\mathbb R}^3}
H(b - x_1 + v_1 \Delta t) \, \lambda_\mu \, f_\mu({\mathbf v})
\, \mbox{d} {\mathbf v}
\nonumber\\
&=&
\frac{\lambda_\mu}{4 \pi \sigma_\mu^2} \,
\int_{(x_1-b)/{\Delta t}}^\infty 
\left( \int_{{\mathbb R}^2}
\delta \left( \sqrt{v_1^2+v_2^2+v_3^2} - \sigma_\mu \right)
\, \mbox{d} v_2\, \mbox{d} v_3\, \right) \mbox{d} v_1
\nonumber\\
&=&
\frac{\lambda_\mu}{2 \sigma_\mu} \,
\left( \sigma_\mu - \frac{x_1-b}{\Delta t} \right)_+
\label{rofixedspeed}
\end{eqnarray}
where $(\cdot)_+$ denotes a positive part. Integrating this formula
over ${\mathbf x}$ in the semi-infinite cuboid 
$(b,\infty)  \times (0,1)^2$, we obtain
$$
\int_{(b,\infty)  \times (0,1)^2}
\varrho({\mathbf x};\Delta t,b)
\, \mbox{d} {\mathbf x}
=
\frac{\lambda_\mu}{2 \sigma_\mu} \,
\int_b^{\infty} 
\left( \sigma_\mu - \frac{x_1-b}{\Delta t} \right)_+
\, \mbox{d} x_1
$$
$$
=
\frac{\lambda_\mu}{2 \sigma_\mu} \,
\int_0^{\sigma_\mu \Delta t} 
\left( \sigma_\mu - \frac{x_1}{\Delta t} \right)
\, \mbox{d} x_1
=
\frac{\lambda_\mu \sigma_\mu \Delta t}{4}.
$$ 
Substituting (\ref{lambda3Dfixedspeed}) for $\lambda_\mu$ and 
(\ref{fvel3Dfixedspeed}) for $\sigma_\mu$, we obtain
(\ref{averpartdeltatBC}).

\end{proof}

Lemma \ref{lemmaboundaryBC} can be used to specify boundary conditions 
for simulations of the MD models [B] and [C] in finite domains as we
did for the one-dimensional case in Lemma \ref{lemmaboundary}.
In Section \ref{secfromMDtoBD3D}, we will use Lemma \ref{lemmaboundaryBC}
to develop and analyse multiscale approaches which can efficiently 
and accurately compute results with an MD-level of detail in a 
(relatively small) subdomain $\Omega_D \subset \Omega$ 
by using coarser BD simulations in the remainder. The geometry of
the desired multiscale method is formulated using (\ref{geom1})
where an MD model is used in $\Omega_D$, a coarser BD model is 
used in $\Omega_C$ and these models are coupled across the interface $I$. 
The situation is schematically shown in Figure \ref{figure2}(d) which 
presents a two-dimensional version of our multiscale set up. Here, 
blue point particles describe heat bath molecules which are used 
in $\Omega_D$. Large biomolecules of interest are denoted as 
grey circles. They are simulated using BD in $\Omega_C$. 
The red line denotes interface $I$. 

\begin{figure}[t]
\centerline{
\raise 1cm \hbox{\epsfig{file=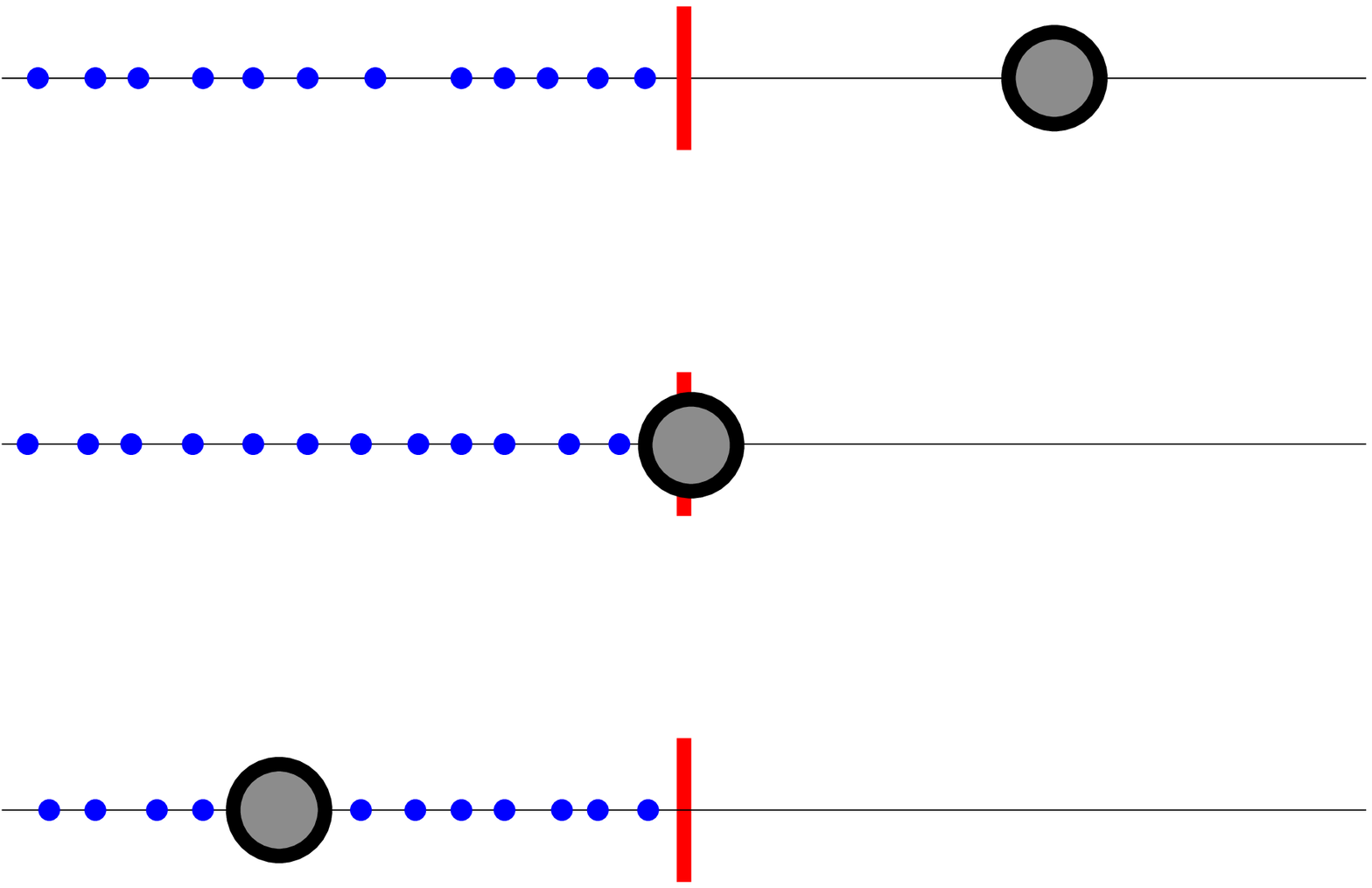,height=3.3cm}}
\epsfig{file=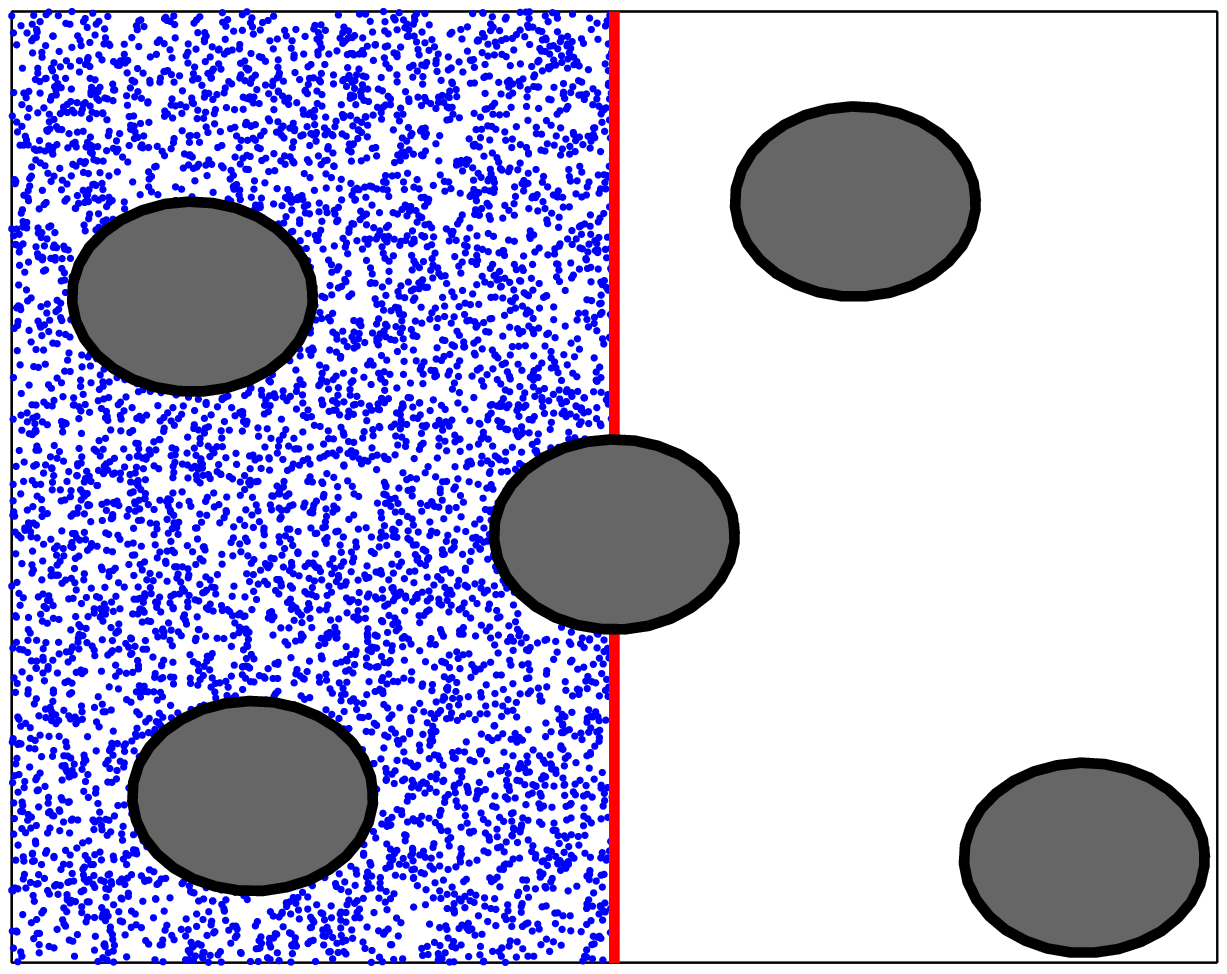,height=5cm}
}
\vskip -53 mm 
\noindent (a) \hskip 57 mm (d) \hskip  10mm 
$\Omega_D$ \hskip 17mm $\Omega_C$\par
\vskip 0.1mm
\noindent \hskip  18mm $\Omega_D$ 
\hskip 8mm $I$ \hskip 9mm $\Omega_C$
\vskip 10mm
\noindent (b)
\vskip 10mm
\noindent (c)
\vskip 2mm
\hskip 95.5mm $I$ \par
\vskip 3 mm
\caption{(a)--(c) {\it Schematic of one-dimensional multiscale
set up $(\ref{geom1})$.} (d) {\it Schematic of multiscale set up 
$(\ref{geom1})$ (in two dimensions)}.}
\label{figure2}
\end{figure}

The schematic in Figure \ref{figure2}(d) is presented in two spatial
dimensions to better visualize the problem geometry. MD models [B] 
and [C] are formulated in a three-dimensional physical space. In the
three-dimensional version of Figure \ref{figure2}(d),
the cloud of blue particles would cover grey ball. To get some insights
into this multiscale problem, we start with the one-dimensional
MD model [A].

\section{From one-dimensional MD model [A] to Brownian dynamics}
\label{secfromMDtoBD1D}

In the case of one-dimensional MD model [A], the situation is schematically 
shown in Figures \ref{figure2}(a)-(c) where we only consider one large 
(heavy) particle, i.e. $N=1$. The large particle can either be
in $\Omega_C$ (see Figure \ref{figure2}(a)), or in $\Omega_D$ 
(see Figure \ref{figure2}(c)) or crossing the boundary as it is shown 
in Figure \ref{figure2}(b). Our geometry is given by
(\ref{geom1}) where
$$
\Omega = (-L,L),
\qquad
\Omega_D = (-L,0),
\qquad
\Omega_C = (0,L),
\qquad
\mbox{and}
\qquad
I = \{0\}.
$$
The large particle covers the interval $(X(t)-R,X(t)+R)$. Let us consider
that the large particle intersects the interface $I$
as it is shown in Figure \ref{figure2}(b). Then
$I \subset (X(t)-R,X(t)+R)$ which is equivalent to $X(t) \in (-R,R)$. 
The heat bath particles are simulated in $\Omega_D$ 
using the MD model [A]. Let us choose $\Delta t$
so small that the probability of two collisions happening 
in the time interval $(t,t + \Delta t)$ is negligible.
Since we do not explicitly simulate heat bath particles in 
$\Omega_C$, we will consider an additional correction of the 
velocity of the heavy particle in the form
\begin{equation}
V(t + \Delta t)
=
\widetilde{V}(t + \Delta t)
+
\alpha(V(t)) \, \Delta t
+
\beta(V(t)) \, \sqrt{\Delta t} \, \xi,
\label{halfSDE}
\end{equation}
where $\widetilde{V}(t + \Delta t)$ is the post-collision velocity of 
the heavy particle at time $t + \Delta t$ which only takes into account 
collisions with the heat bath particles from the left. It is either equal 
to $V(t)$ or computed by (\ref{col1Dmu}) if a collision with a heat bath 
particle occurred in $\Omega_D$. Equation (\ref{halfSDE}) is adding both 
drift term $\alpha(V(t)) \, \Delta t$ and noise term 
$\beta(V(t)) \, \sqrt{\Delta t} \, \xi$ where
$\xi$ is a normally distributed random number with zero mean
and unit variance. The drift and noise terms implicitly take into 
account collisions at the right boundary $(X(t)+R)$ of the heavy 
particle. Passing $\Delta t \to 0$, we
observe that the contributions of the collisions at the right 
boundary are given by the It\={o} stochastic differential equation
\begin{equation}
\mbox{d}V
=
\alpha(V) \, \mbox{d}t
+
\beta(V) \, \mbox{d}W.
\label{halfSDE2}
\end{equation}
If we explicitly modelled heat bath particles in $\Omega_C$, then they 
would be distributed according to the Poisson distribution with density 
$\lambda_\mu$ in the interval $(X(t)+R,\infty)$. Their initial velocities 
would be given according to $(\ref{distr1Dv})$. Thus, using Lemma 
\ref{lemmaboundary} and (\ref{col1Dmu}), we can estimate
the drift coefficent of the stochastic differential
equation (\ref{halfSDE2}) to get
$$
\alpha(V)
=
\frac{1}{\Delta t}
\int_{X(t)+R}^\infty
\int_{-\infty}^{\infty}
\frac{2(v-V) }{\mu + 1} 
\,
H \left( \frac{X(t) + R - x}{\Delta t} + V- v \right) \, 
\lambda_\mu 
f_\mu(v) 
\, \mbox{d} v
\, \mbox{d} x,
$$
where $H(\cdot)$ is the Heaviside step function. Using (\ref{distr1Dv}),
we obtain
\begin{eqnarray}
\alpha(V)
&=&
\frac{2 \lambda_\mu}{\Delta t (\mu + 1)\sigma_\mu \sqrt{2 \pi}}
\int_{0}^{\infty}
\int_{-\infty}^{V - x/ \Delta t}
(v-V) \, \exp \left( - \frac{v^2}{2 \sigma_\mu^2} \right)
\, \mbox{d} v
\, \mbox{d} x
\nonumber
\\
&=&
-
\frac{\lambda_\mu}{\mu + 1}
\left(
\big(
\sigma_\mu^2 + V^2
\big) 
\,
\mbox{erfc} 
\left[ 
- \frac{V}{\sigma_\mu \sqrt{2}}
\right]
+
\frac{V \,\sigma_\mu \sqrt{2}}{\sqrt{\pi}}
\exp
\left[
-\frac{V^2}{2 \, \sigma^2_\mu}
\right]
\right),\qquad
\label{alphaprec}
\end{eqnarray}
where $\lambda_\mu$ and $\sigma_\mu$ are given by
(\ref{distr1Dx}) and (\ref{distr1Dv}). In the limit
$\mu \to \infty$, we have $V/\sqrt{\mu+1} \to 0$. Thus we
use the Taylor expansion in (\ref{alphaprec}) to get
\begin{equation}
\alpha(V)
\approx
-
\frac{\gamma \sqrt{\pi (\mu+1) D \gamma}}{4 \sqrt{2}}
-
\frac{\gamma}{2} \, V
-
\frac{\sqrt{\pi \gamma}}{4 \sqrt{2 D (\mu+1)}}
\, V^2. 
\label{alphatrunc}
\end{equation}
The noise term in (\ref{halfSDE2}) can be computed by
$$
\beta^2(V)
=
\frac{1}{\Delta t}
\int_{X(t)+R}^\infty
\int_{-\infty}^{\infty}
\frac{4(v-V)^2}{(\mu + 1)^2} 
\,
H \left( \frac{X(t) + R - x}{\Delta t} + V- v \right) \, 
\lambda_\mu 
f_\mu(v) 
\, \mbox{d} v
\, \mbox{d} x.
$$
Using (\ref{distr1Dv}), we obtain
$$
\beta^2(V)
=
\frac{4 \lambda_\mu}{\Delta t (\mu + 1)^2 \sigma_\mu \sqrt{2 \pi}}
\int_{0}^{\infty}
\int_{-\infty}^{V - x/ \Delta t}
(v-V)^2 \, \exp \left( - \frac{v^2}{2 \sigma_\mu^2} \right)
\, \mbox{d} v
\, \mbox{d} x
\qquad
\qquad
\quad
$$
$$
=
\frac{2 \lambda_\mu}{(\mu + 1)^2}
\left(
V (3 \sigma^2_\mu + V^2)
\, \mbox{erfc} 
\left[ 
- \frac{V}{\sigma_\mu \sqrt{2}}
\right]
+
\frac{2 (2 \sigma^2_\mu + V^2) \sigma_\mu}{\sqrt{2 \pi}}
\exp \left[ - \frac{V^2}{2 \, \sigma_\mu^2} \right]
\right).
$$
Using (\ref{distr1Dx}), (\ref{distr1Dv}) and the Taylor expansion, we obtain
\begin{equation}
\beta(V)
\approx
\sqrt{
\gamma^2 D
+
\frac{3 \gamma \sqrt{\pi D \gamma}}{2 \sqrt{2(\mu+1)}}
\, V 
+
\frac{3 \gamma}{2 (\mu + 1)}
\, V^2}.
\label{betatrunc}
\end{equation}
Equations (\ref{alphatrunc}) and (\ref{betatrunc}) are used in the
multiscale algorithm in Table \ref{tablealgH1}.
\begin{table}
\framebox{%
\hsize=0.97\hsize
\vbox{
\leftskip 10mm
\parindent -10mm
[M1] \hskip 1.2mm
Compute ``free-flight positions'' of heat bath particles 
and the heavy particle at time $t+\Delta t$ using the step [A1]. 

\smallskip

[M2] \hskip 1.2mm
Compute post-collision velocities by (\ref{col1Dmu})
for every pair of particles which collided using the step [A2].

\smallskip

[M3] \hskip 1.2mm
Terminate trajectories of heat bath particles which 
left the subdomain $\Omega_D = (-L,0)$. Update $n$ accordingly. 

\smallskip

[M4] \hskip 1.2mm
Implement the influx of heat bath particles through the boundary $x=-L$
using the step [A4].

\smallskip

[M5] \hskip 1.2mm
If $X(t) \not \in (-R,R)$, then generate a random number $r_2$ uniformly 
distributed in $(0,1)$. If $r_2<\gamma (\mu + 1) \Delta t/8$, then
increase $n$ by 1, and introduce a new heat bath particle at position
$x^n(t+\Delta t)$ with velocity $v^n(t+\Delta t)$ which are 
sampled according to probability distributions 
(\ref{samplingnewXrightvnule}) and (\ref{samplingnewVrightvnule}).

\smallskip

[M6] \hskip 1.2mm
If $X(t) \in (-R,R)$, then update the heavy particle velocity using
(\ref{halfSDE}).

\smallskip

[M7] \hskip 1.2mm
If $X(t) \in [R,L)$, then update the velocity of the heavy particle 
using (\ref{SDEvel}).

\smallskip

[M8] \hskip 1.2mm
Continue with the step [M1] using time $t=t+\Delta t$.

\par \vskip 0.8mm}
}\vskip 1mm
\caption{One iteration of the computer implementation of the multiscale
algorithm which is based on the MD model [A].}
\label{tablealgH1}
\end{table} 
The first two steps [M1] and [M2] are the same as [A1] and [A2].
Since heat bath particles are only simulated in the subdomain
$\Omega_D = (-L,0)$, we remove all particles which 
left $\Omega_D$ during the time interval $(t, t+\Delta t)$ in the step [M3].
The step [M4] is the same as [A4] which introduces  
heat bath particles which have entered $\Omega_D$ through its
left boundary $x=-L$ during the time interval $(t, t+\Delta t)$.
The boundary at $x=0$ is treated in the step [M5] if the heavy particle
does not intersect with this boundary. We assume that $\Delta t$ 
is chosen so small that (\ref{averpartdeltat}) is much smaller than 1.
Then (\ref{averpartdeltat}) can be interpretted as a probability of 
introducing one particle from the left (resp. right) during one timestep. 
A particle introduced close to the right boundary of $\Omega_D$ in the step 
[M5] will have its position sampled according 
to the probability distribution
\begin{equation}
C_1 \, \mathrm{erfc} 
\left(\frac{-x}{\Delta t \, \sigma_\mu \sqrt{2}} \right),
\qquad
\mbox{for}
\;\;
x \in (-\infty,0),
\label{samplingnewXrightvnule}
\end{equation}
where $C_1$ is a normalization constant. The probability 
distribution (\ref{samplingnewXrightvnule}) can be justified 
using the same argument as Lemma \ref{lemmaboundary}
and equation (\ref{samplingnewXright}). 
To sample from the probability distribution (\ref{samplingnewXrightvnule}), 
we again use the acceptance-rejection algorithm in 
Table \ref{tableerfcsampling} with parameters $a_1$ and
$a_2$ given by (\ref{valuesa1a2}). In the step 
[M5], we also sample the velocity 
$v \in {\mathbb R}$ of the new particle using the truncated 
Gaussian distribution
\begin{equation}
C_2 \, H(x - v \Delta t) \, f_\mu(v),
\label{samplingnewVrightvnule}
\end{equation}
where $C_2$ is a normalization constant.
To sample random numbers according to the truncated normal 
distributions in the steps [M4] and [M5], we again use the
acceptance-rejection algorithm which is derived as
Proposition 2.3 in \cite{Robert:1995:STN}.
If the heavy particle does intersect with the boundary $I$, then 
the step [M6] is executed. It uses (\ref{halfSDE}) to incorporate 
collisions of heat bath particles from the right.
If the particle does not intersect with $\Omega_D$, then
we simulate it in the step [M7] using the discretized 
version of (\ref{BDVeq}) given by
\begin{equation}
V(t + \Delta t)
=
- \gamma V(t) \, \Delta t
+
\gamma \, \sqrt{2 \, D \, \Delta t} \, \xi,
\label{SDEvel}
\end{equation}
In Figure \ref{figure3}, we present illustrative results 
\begin{figure}[t]
\picturesAB{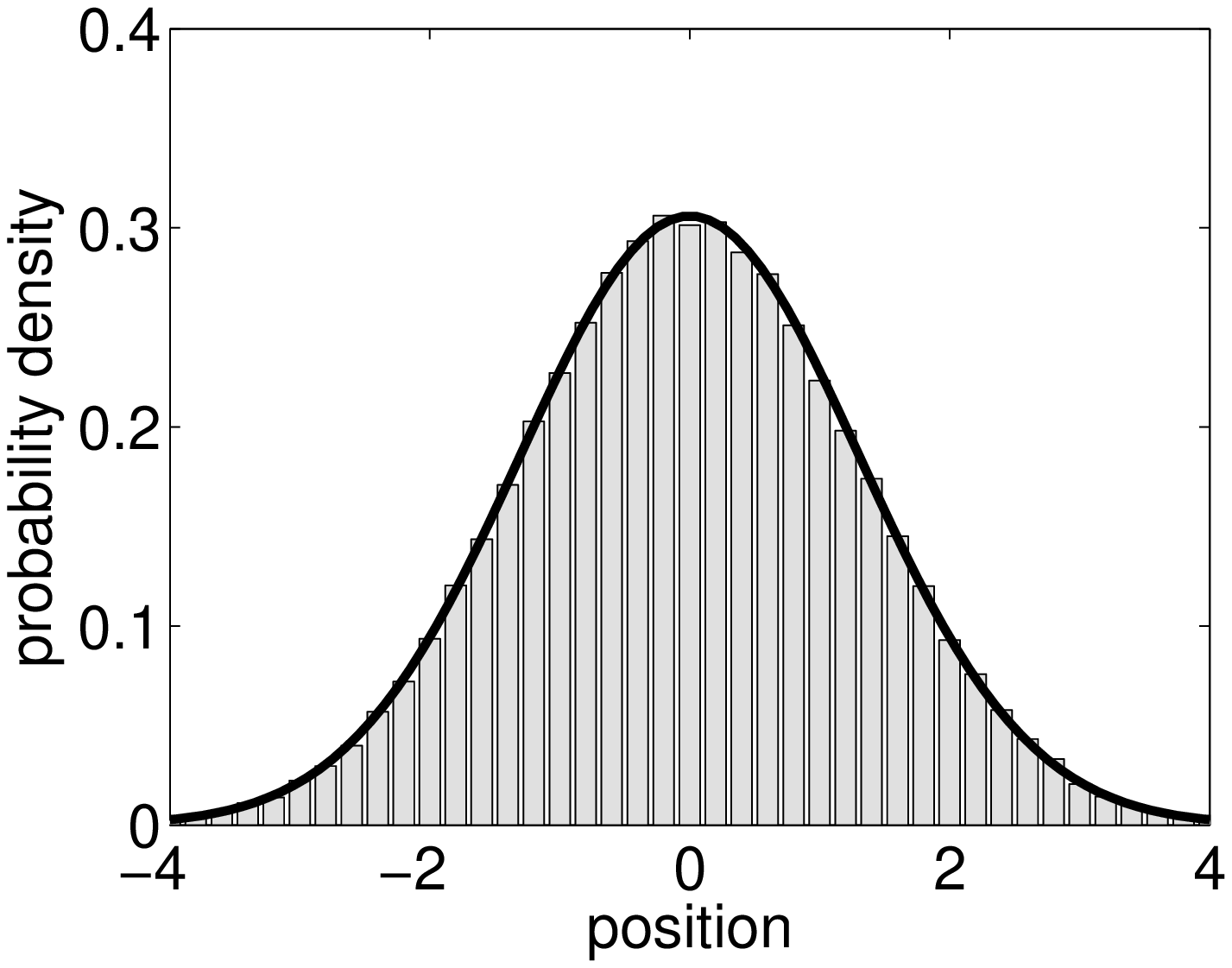}{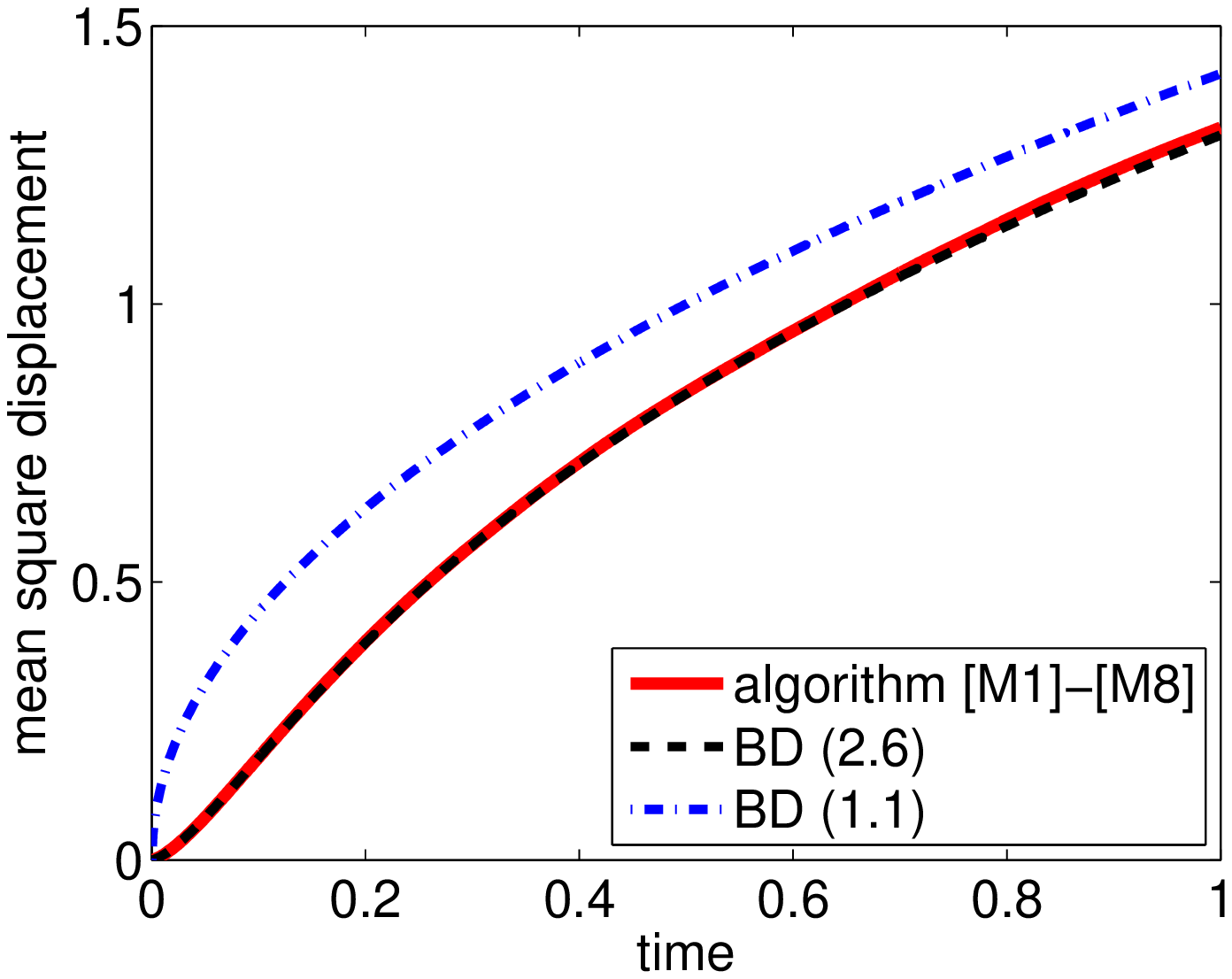}{5.3cm}{3mm}
\caption{(a) {\it Probability distribution of the heavy particle 
at time $t=1$ computed by the multiscale algorithm 
{\rm [M1]}--{\rm [M8]} (gray histogram) is compared with the
distribution $(\ref{shiftedGaussion})$ given by the BD model 
$(\ref{evollimitXV})$ (black solid line).
(b) {\it The time evolution of the mean square displacement computed 
by $10^4$ realizations of the algorithm {\rm [M1]}--{\rm [M8]}
(red solid line) is compared with equation 
$(\ref{eqmsd15})$ (black dashed line) and 
$\sqrt{2 D \, t}$ (blue dot-dashed line). \hfill\break
We use $\mu = 10^3,$ $\gamma = 10$, $D=1$, $\Delta t = 10^{-7}$,
$L=10$, $R=1$, $X(0) = 0$ and $V(0) = 0$.}
}}
\label{figure3}
\end{figure}
computed by the algorithm [M1]-[M8]. We consider one heavy particle
which starts at position $X(0) = 0$ with velocity $V(0) = 0$ as we did
in Figure \ref{figure1}. The distribution of its position at time
$t=1$, computed using $10^5$ realizations of the algorithm
[M1]-[M8], is plotted in Figure \ref{figure3}(a). It is compared
with the distribution obtained by the limiting BD model 
(\ref{evollimitXV}) which is, for $t \gg \gamma^{-1}$, 
given by \citep{Hagan:1989:MET}
\begin{equation}
\frac{1}{\sqrt{4 \pi D \left( t - t^* \right)}}
\exp \left[ - \frac{x^2}{4 D \left( t - t^* \right) } \right],
\qquad \mbox{where} \quad t^* = \frac{3}{2\gamma}.
\label{shiftedGaussion}
\end{equation}
In Figure \ref{figure3}(b), we plot the time evolution of the
mean square displacement. This figure can be directly compared with 
Figure \ref{figure1}(b), because we use the same parameter values.
The results computed by the multiscale algorithm [M1]-[M8] compare 
well with the results given by the BD model (\ref{evollimitXV}). We 
have already shown in Figure \ref{figure1}(b) that the limiting BD model
(\ref{evollimitXV}) also compares well with the MD simulations. 
In particular, the algorithm [M1]-[M8] is able to compute results with the 
MD-level precision by using coarser BD models in a part of the computational
domain.

The stochastic differential equation (\ref{halfSDE2}) was derived
for collisions from the right. Using the same argument, we can
also derive a stochastic differential equation which is approximating
the effect of collisions from the left. We obtain
\begin{equation}
\mbox{d}V
=
- \alpha(-V) \, \mbox{d}t
+
\beta(-V) \, \mbox{d}W.
\label{halfSDE3}
\end{equation}
Adding (\ref{halfSDE2}) and (\ref{halfSDE3}) and using the
independence of noise terms in (\ref{halfSDE2}) and (\ref{halfSDE3}), 
we can approximate collisions from both sides by the following SDE for
the velocity of the heavy particle:
\begin{equation}
\mbox{d}V
=
\big( \alpha(V)- \alpha(-V) \big) \, \mbox{d}t
+
\sqrt{ \beta^2(V) + \beta^2(-V) } \, \mbox{d}W.
\label{fullSDE3}
\end{equation}
Substituting (\ref{alphatrunc}) and (\ref{betatrunc}), we derive
(\ref{BDVeq}). In particular, we have verified the limiting result
in Lemma \ref{convergenceelastic}. 

\section{From three-dimensional MD models [B] and [C] to Brownian dynamics}
\label{secfromMDtoBD3D}

We use a simple multiscale geometry where domain $\Omega = {\mathbb R}^3$ 
is divided into two half spaces. Heavy molecules are simulated
in both half spaces. In $\Omega_D = (-\infty,0) \times {\mathbb R}^2$,
we use the MD model [B] or [C]. It is coupled with the BD model
given by (\ref{BDXeq})--(\ref{BDVeq}) in 
$\Omega_C = (0,\infty) \times {\mathbb R}^2.$ This set up is a
three-dimensional version of multiscale problems which are
schematically drawn in Figure \ref{figure2}. Boundary conditions 
for heat bath particles at the interface $I = \{0\} \times {\mathbb R}^2$
can be specified using Lemma \ref{lemmaboundaryBC}.   

As in Section \ref{secfromMDtoBD1D}, we need to analyse the behaviour
of a heavy molecule when it intersects with the interface $I$. 
Such molecule is subject to the collisions 
with heat bath particles on the part of its surface which lies in 
$\Omega_D$. This has to be compensated by using a suitable random 
force from $\Omega_C$, so that the overall model is equivalent 
to (\ref{BDXeq})--(\ref{BDVeq}) in the BD limit. To simplify the
presentation of the algorithm, we use the same time step 
in $\Omega_D$ and $\Omega_C$. In Section \ref{bioapplication} we present
coupling of three-dimensional MD models with the 
BD model (\ref{BDSDE}) which will make use of different time
steps in different parts of the computational domain.

The heavy particle is the ball with centre ${\mathbf X} = [X_1, X_2, X_3]$ 
with velocity ${\mathbf V} = [V_1,V_2,V_3]$ and radius $R$. It intersects 
the interface $I$ if $X_{1} (t) \in (-R,R)$. Let us consider that heat bath 
particles are simulated in $\Omega_D$ using the MD model [B] or the
MD model [C]. Let us choose $\Delta t$ so small that the probability of two 
collisions happening in the time interval $(t,t + \Delta t)$ is negligible.
Since we do not explicitly simulate the heat bath particles in 
$\Omega_C$, we will consider an additional correction of the 
velocity of the heavy particle in the form
\begin{equation}
{\mathbf V}(t + \Delta t)
=
\widetilde{{\mathbf V}}(t + \Delta t)
+
{\boldsymbol \alpha}({\mathbf X}(t),{\mathbf V}(t)) \, \Delta t
+
{\boldsymbol  \beta}({\mathbf X}(t),{\mathbf V}(t)) \, \sqrt{\Delta t} \, 
{\boldsymbol  \xi},
\label{halfSDE3D}
\end{equation}
where $\widetilde{{\mathbf V}}(t + \Delta t)$ is the post-collision 
velocity of the heavy particle at time $t + \Delta t$ which only takes 
into account collisions with the heat bath particles from $\Omega_D$. 
It is either equal to ${\mathbf V}(t)$ or computed by 
(\ref{eq11})--(\ref{eq12}) if a collision with a heat bath 
particle occurred in $\Omega_D$. Note that we dropped the subscript
$\mu$ in (\ref{eq11})--(\ref{eq12}) to simplify our notation.
Equation (\ref{halfSDE3D}) is a generalization of (\ref{halfSDE}) 
to three-dimensional simulations where 
${\boldsymbol \alpha}({\mathbf X}(t),{\mathbf V}(t)) \, \Delta t$
is the drift vector and 
${\boldsymbol  \beta}({\mathbf X}(t),{\mathbf V}(t)) \, \sqrt{\Delta t} \, 
{\boldsymbol \xi}$ is the noise term and
${\boldsymbol \xi} = [\xi_1, \xi_2, \xi_3]$
is the vector of three normally distributed random numbers with zero mean
and unit variances. Passing $\Delta t \to 0$, we observe that the 
contributions of the collisions from $\Omega_D$ are given by the 
It\={o} stochastic differential equation
\begin{equation}
\mbox{d}{\mathbf V}
=
{\boldsymbol \alpha}({\mathbf X}(t),{\mathbf V}(t)) \, \mbox{d}t
+
{\boldsymbol  \beta}({\mathbf X}(t),{\mathbf V}(t)) \, \mbox{d}{\mathbf W}.
\label{halfSDE23D}
\end{equation}
To estimate drift ${\boldsymbol \alpha}$ and diffusion coefficient
${\boldsymbol  \beta}$, we separately consider MD models [B] and [C] 
in the following two subsections.

\subsection{MD model {\rm [B]}}

\noindent
The following lemma will be useful to estimate the drift coefficient
${\boldsymbol \alpha}$.

\begin{lemma}\label{3Dlemmaypoint}
Let $\gamma>0$, $D>0$, $R>0$ and $\Delta t > 0$. Let us consider 
the MD model {\rm [B]} where the positions and velocities of heat bath 
particles are distributed according to $(\ref{lambda3Dexp})$
and $(\ref{fvel3Dexp})$. Let us consider one heavy molecule 
in such a heat bath, i.e. $N=1$, with the position
of its centre to be at ${\mathbf X}(t) = [X_1(t), X_2(t), X_3(t)]$ 
and with velocity ${\mathbf V}(t) = [V_1(t),V_2(t),V_3(t)]$.
Let ${\mathbf y} = (y_1,y_2,y_3)$ be a given point on the surface of
the heavy molecule at time $t$, i.e.
\begin{equation}
(y_1 - X_1(t))^2 + (y_2 - X_2(t))^2 + (y_3 - X_3(t))^2 = R^2.
\label{defsurfacepoint}
\end{equation}
Then the average change of the $j$-th component of the
velocity of the heavy molecule caused by collisions with heat
bath particles in the time interval $(t,t+\Delta t)$ at the
surface area $({\mathbf y},{\mathbf y} + \mbox{{\rm d}}{\mathbf y})$ 
is $\psi_j({\mathbf y}) \, \mbox{{\rm d}}{\mathbf y}$ where
\begin{eqnarray}
\psi_j({\mathbf y})
&=&
- \;
\frac{\lambda_\mu \, \sigma_\mu^2 \,
(y_j -X_j(t)) \, \Delta t }{(\mu + 1) \, R }
+
\frac{4 \, \lambda_\mu \, \sigma_\mu \,
(y_j -X_j(t)) \, \Delta t}{(\mu + 1) \, R^2 \, \sqrt{2 \, \pi}}
\; {\mathbf V}(t) \cdot ({\mathbf y}-{\mathbf X}(t))
\nonumber \\
&&
+ \;
\mathrm{O}\left( \parallel {\mathbf V} \parallel^2 \right).
\label{psijformula} 
\end{eqnarray}
\end{lemma}

\noindent
\begin{proof} Let us consider that a heat bath particle
which was at point ${\mathbf x}$ at time $t$ collided
with the heavy molecule at time $t + \tau \in (t, t+\Delta t)$
at the surface point which had coordinate ${\mathbf y}$ at time $t$.
Then the coordinate of the surface point at the collision
time $t + \tau$ was ${\mathbf y} + \tau {\mathbf V}(t)$
and the pre-collision velocity of the heat bath molecule was 
${\mathbf v} = {\mathbf V}(t) + ({\mathbf y} - {\mathbf x})/\tau$.
Using equation (\ref{eq11}), we can write the change
of the velocity of the heavy molecule during the collision 
as
\begin{equation}
\frac{2}{\mu + 1}
\, 
\left[{\mathbf v} - {\mathbf V}(t) \right]^\perp
=
\frac{2}{\mu + 1}
\, 
\left(
\frac{({\mathbf y} - {\mathbf x})}{\tau} \cdot 
\frac{({\mathbf y}-{\mathbf X}(t))}{R} 
\right)
\frac{({\mathbf y}-{\mathbf X}(t))}{R}.
\label{paramtheta}
\end{equation}
The position ${\mathbf x}$ of the heat bath particle must
be in the half space which lies above the plane
tangent to the heavy molecule at the collision point 
${\mathbf y} + \tau V(t)$. It can be parametrized by 
$$
{\mathbf x} = {\mathbf y} + \tau \, {\mathbf V}(t) 
+ c_1 \, \tau \, \frac{({\mathbf y}-{\mathbf X}(t))}{R}
+ c_2 \, \tau \, {\boldsymbol \eta_2} + c_3 \, \tau \, {\boldsymbol \eta_3},
$$
where $c_1 > 0,$ $c_2 \in {\mathbb R},$ $c_3  \in {\mathbb R}$,
and $({\mathbf y}-{\mathbf X}(t))/R$, ${\boldsymbol \eta_2}$, 
${\boldsymbol \eta_3}$ is the orthornormal basis in ${\mathbb R}^3.$ 
Then (\ref{paramtheta}) reads as follows
\begin{equation*}
\frac{2}{\mu + 1}
\, 
\left[{\mathbf v} - {\mathbf V}(t) \right]^\perp
=
-
\frac{2}{(\mu + 1) \, R}
\, 
\left(
c_1
+
{\mathbf V}(t) 
\cdot 
\frac{({\mathbf y}-{\mathbf X}(t))}{R} 
\right)
({\mathbf y}-{\mathbf X}(t)).
\end{equation*}
Thus we have
\begin{eqnarray}
\psi_j({\mathbf y})
&=&
-
\frac{2 \, \lambda_\mu \,
(y_j -X_j(t))}{(\mu + 1) \, R}
\int_{0}^\infty
\int_{-\infty}^\infty
\int_{-\infty}^\infty
\int_{0}^{\Delta t}
\left(
c_1 + {\mathbf V}(t) \cdot 
\frac{({\mathbf y}-{\mathbf X}(t))}{R} 
\right)^2
\,
\nonumber
\\
&\times&
f_\mu \left( 
- c_1 \, \frac{{\mathbf y}-{\mathbf X}(t)}{R}
- c_2 \, {\boldsymbol \eta_2} 
- c_3 \, {\boldsymbol \eta_3}
\right) 
\, \mbox{d} \tau
\, \mbox{d} c_3
\, \mbox{d} c_2
\, \mbox{d} c_1.
\label{psijformulaaux}
\end{eqnarray}
Substituting (\ref{fvel3Dexp}) for $f_\mu$ and
integrating over $\tau$, $c_2$ and $c_3$, we have
\begin{equation*}
\psi_j({\mathbf y})
=
-
\frac{\lambda_\mu \, (y_j -X_j(t)) \, 
\Delta t \, \sqrt{2}}{(\mu + 1) \, R \, \sigma_\mu \, \sqrt{\pi}}
\int_{0}^\infty \!\!
\left(
c_1 + {\mathbf V}(t) \cdot 
\frac{({\mathbf y}-{\mathbf X}(t))}{R} 
\right)^2
\,
\exp \left[ 
- \frac{c_1^2}{2 \sigma_\mu^2} \right]
\, \mbox{d} c_1.
\end{equation*}
Integrating over $c_1$, we deduce (\ref{psijformula}). 
\end{proof}

\noindent
Using Lemma \ref{3Dlemmaypoint}, we can compute the drift
coefficient ${\boldsymbol \alpha}({\mathbf X}(t),{\mathbf V}(t))$ 
in equation (\ref{halfSDE23D}) as follows
\begin{equation}
\alpha_j({\mathbf X}(t),{\mathbf V}(t))
=
\frac{1}{\Delta t}
\int_{S({\mathbf X}(t))}
\psi_j({\mathbf y})
\mbox{d} {\mathbf y}
\label{drift3Dcoefficient}
\end{equation}
where $S({\mathbf X}(t))$ is the part of the surface of the heavy molecule
which intersects the BD subdomain $\Omega_C$, i.e.
$$
S({\mathbf X}(t)) =
\left\{ {\mathbf y} \in \Omega_C \,\big|\, 
{\mathbf y} \; \mbox{satisfies} \; (\ref{defsurfacepoint})
\right\}.
$$
Substituting (\ref{psijformula}) into (\ref{drift3Dcoefficient}), we have
\begin{eqnarray*}
\alpha_j({\mathbf X}(t),{\mathbf V}(t))
&=&
-
\frac{\lambda_\mu \, \sigma_\mu^2}{(\mu + 1) \, R }
\int_{S({\mathbf X}(t))}
(y_j -X_j(t))
\mbox{d} {\mathbf y}
\\
&-&
\frac{4 \, \lambda_\mu \, \sigma_\mu \, V_j(t)}{(\mu + 1) 
\, R^2 \, \sqrt{2 \, \pi}}
\int_{S({\mathbf X}(t))}
(y_j -X_j(t))^2
\mbox{d} {\mathbf y}.
\end{eqnarray*}
Using (\ref{lambda3Dexp}) and (\ref{fvel3Dexp}) and evaluating
the surface integrals, we obtain
\begin{equation}
\alpha_1({\mathbf X},{\mathbf V})
=
-
\frac{3 \, \gamma \, 
\sqrt{\pi \, (\mu+1) \, D \, \gamma}}{8 \, \sqrt{2}}
\left(1 - \frac{X_1^2}{R^2} \right)
-
\frac{\gamma \, V_j}{2}
\, 
\left
(1 + \frac{X_1^3}{R^3}
\right)
\label{alpha1exp}
\end{equation}
and
\begin{equation}
\alpha_j({\mathbf X},{\mathbf V})
=
-
\frac{\gamma \, V_j}{4}
\, 
\left(2 + 3 \frac{X_1}{R} - \frac{X_1^3}{R^3} \right), 
\qquad \mbox{for} \; j = 2, \; 3,
\label{equationalpha2alpha3}
\end{equation}
where we dropped the dependence on time $t$ to shorten the 
resulting formulae.
The noise matrix ${\boldsymbol  \beta}({\mathbf X}(t),{\mathbf V}(t))$
will be estimated using ${\boldsymbol  \beta}({\mathbf X}(t),{\mathbf 0})$, 
i.e. we will only use the first term in the Taylor expansion 
in ${\mathbf V}$. Using similar arguments as in the proof 
of (\ref{psijformulaaux}) and (\ref{drift3Dcoefficient}), we have
\begin{eqnarray}
\beta^2_{i,i}({\mathbf X}(t),{\mathbf 0})
&=&
-
\frac{4 \, \lambda_\mu}{(\mu + 1)^2 \, R^2}
\int_{S({\mathbf X}(t))}
\int_{0}^\infty
\int_{-\infty}^\infty
\int_{-\infty}^\infty
c_1^3 \, (y_i -X_i(t))^2
\nonumber 
\\
&\times&
f_\mu \left( 
- c_1 \, \frac{{\mathbf y}-{\mathbf X}(t)}{R}
- c_2 \, {\boldsymbol \eta_2} 
- c_3 \, {\boldsymbol \eta_3}
\right) 
\, \mbox{d} c_3
\, \mbox{d} c_2
\, \mbox{d} c_1
\mbox{d} {\mathbf y},
\qquad
\label{beta23Daux}
\end{eqnarray}
for $i=1,2,3$. Substituting (\ref{fvel3Dexp}) for $f_\mu$,
(\ref{lambda3Dexp}) for $\lambda_\mu$ and
using ${\boldsymbol  \beta}({\mathbf X},{\mathbf V}) 
= {\boldsymbol  \beta}({\mathbf X},{\mathbf 0})$, 
we obtain
\begin{eqnarray}
\beta_{1,1}({\mathbf X},{\mathbf V}) &=& 
\gamma \, \sqrt{D} \, \sqrt{ 1 + \frac{X_1^3}{R^3} } \, , 
\nonumber
\\
\beta_{j,j}({\mathbf X},{\mathbf V}) &=& 
\gamma \, \sqrt{D} \,
\sqrt{ 1 + \frac{3 X_1}{2 R} - \frac{X_1^3}{2 R^3} } \, ,
\qquad \mbox{for} \; j = 2, \; 3 \, ,
\label{noiseterm3D}
\\
\beta_{i,j}({\mathbf X},{\mathbf V}) &=& 0 \, , 
\qquad\qquad\qquad\qquad\qquad\qquad\mbox{for} \; i \ne j ,
\nonumber
\end{eqnarray}
where the last equation can be verified using the same argument
as equation (\ref{beta23Daux}). Notice that by
substituting $X_1 = R$ into (\ref{alpha1exp}), 
(\ref{equationalpha2alpha3}) and (\ref{noiseterm3D}) 
we verify the limiting result in Lemma \ref{3Dlemmaexp}.

\subsection{MD model {\rm [C]}}

\noindent
Equations (\ref{psijformulaaux}), (\ref{drift3Dcoefficient}) and
(\ref{beta23Daux}) which are derived in the previous section
are applicable to both MD models [B] and [C]. To estimate 
the drift coefficient ${\boldsymbol \alpha}({\mathbf X},{\mathbf V})$ 
for the MD model [C], we substitute (\ref{lambda3Dfixedspeed}) 
for $\lambda_\mu$ and (\ref{fvel3Dfixedspeed}) for $f_\mu$ 
in (\ref{psijformulaaux}) and (\ref{drift3Dcoefficient}). We obtain
\begin{equation}
\alpha_1({\mathbf X},{\mathbf V})
=
-
\frac{\gamma \, 
\sqrt{(\mu+1) \, D \, \gamma}}{2}
\left(1 - \frac{X_1^2}{R^2} \right)
-
\frac{\gamma \, V_j}{2}
\, 
\left
(1 + \frac{X_1^3}{R^3}
\right)
\label{alpha1fixedspeed}
\end{equation}
and $\alpha_2({\mathbf X},{\mathbf V})$ and 
$\alpha_3({\mathbf X},{\mathbf V})$ are again given by 
(\ref{equationalpha2alpha3}). Substituting (\ref{lambda3Dfixedspeed}) 
and (\ref{fvel3Dfixedspeed})
in (\ref{beta23Daux}) and integrating, we obtain 
that noise matrix ${\boldsymbol \beta}({\mathbf X},{\mathbf V})$
satisfies (\ref{noiseterm3D}). We again notice that the
special choice $X_1 = R$ in (\ref{alpha1fixedspeed}) can be used
to verify the limiting result in Lemma \ref{3Dlemmafixedspeed}.

\subsection{Illustrative numerical results}

In the previous two subsections we have observed that the only difference
between MD models [B] and [C] is a different formula for the 
coefficient $\alpha_1({\mathbf X},{\mathbf V})$ in 
(\ref{halfSDE23D}), given by (\ref{alpha1exp}) and
(\ref{alpha1fixedspeed}), respectively. The remaining terms 
in (\ref{halfSDE23D}) are the same, given by (\ref{equationalpha2alpha3}) 
and (\ref{noiseterm3D}). In this section, we present an illustrative
computation with the MD model [C], but the same results can also 
be obtained with the MD model [B] (results not shown). An
illustrative computation with the MD model [B] is presented 
later in Section \ref{bioapplication}.

\newcommand{\picturesABal}[5]{
\centerline{
\hskip #4
\raise #3 \hbox{\raise 0.9mm \hbox{(a)}}
\hskip -8mm
\epsfig{file=#1,height=#3}
\raise #3 \hbox{\raise 0.9mm \hbox{(b)}}
\hskip -8mm
\raise 1.5mm \hbox{\epsfig{file=#2,height=#5}}
}}

We consider a three-dimensional generalization of the illustrative 
problem from Figure \ref{figure3} from Section \ref{secfromMDtoBD1D}. 
One heavy particle which starts at position ${\mathbf X}(0) 
= [0,0,0]$ with velocity 
${\mathbf V}(0)=[0,0,0]$ is simulated using a three-dimensional generalization
of the algorithm [M1]-[M8]. We use the MD model [C] in 
$\Omega_D = (-\infty,0) \times {\mathbb R}^2$ and
the BD model (\ref{BDXeq})--(\ref{BDVeq}) in 
$\Omega_C = (0,\infty) \times {\mathbb R}^2.$
In the step [M6], we replace 
(\ref{halfSDE}) with its three-dimensional analogue
(\ref{halfSDE3D}) where drift ${\boldsymbol \alpha}$ and 
diffusion coefficient ${\boldsymbol  \beta}$
are given by (\ref{alpha1fixedspeed}), 
(\ref{equationalpha2alpha3}) and (\ref{noiseterm3D}).
The distribution of $X_1$ positions of the heavy particle at time
$t=1$, computed using $10^5$ realizations of the 
multiscale algorithm, is plotted in Figure \ref{figure4}(a).
The limiting BD result is again given by  
\begin{figure}[t]
\picturesABal{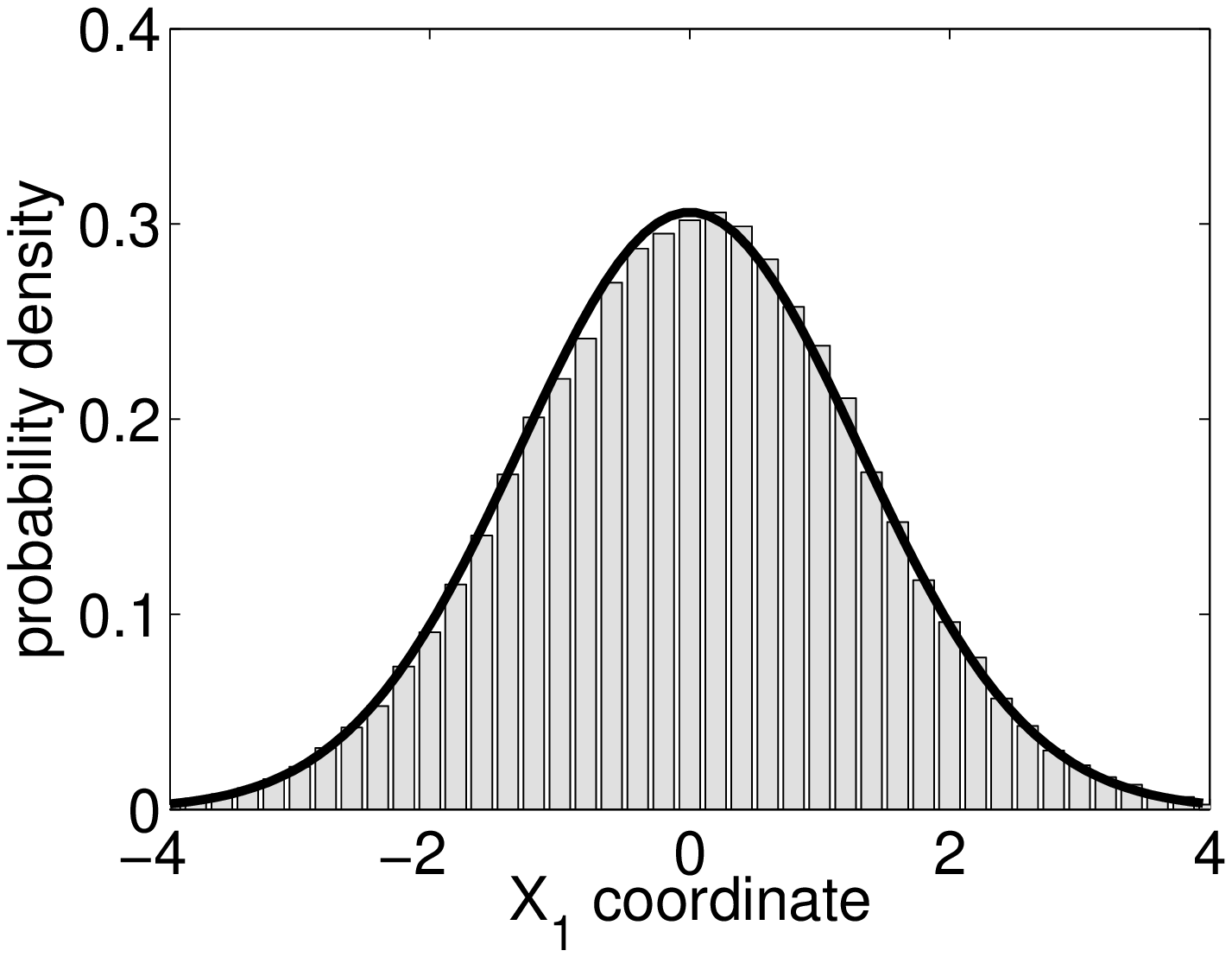}{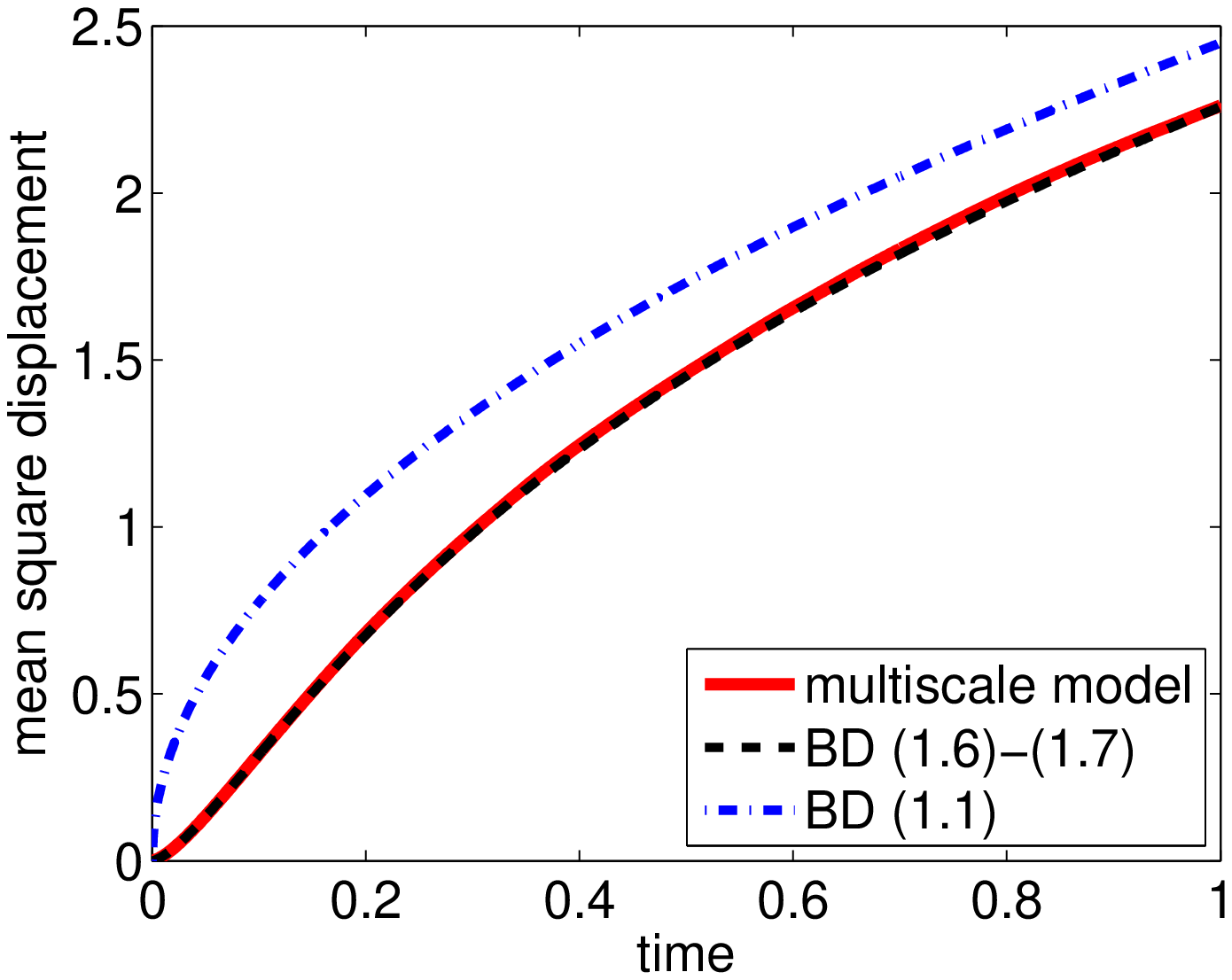}{5.3cm}{3mm}{5.1cm}
\caption{(a) {\it Probability distribution of the first coordinate,
$X_1$, of the heavy particle 
at time $t=1$ computed by the three-dimensional multiscale 
algorithm (gray histogram) is compared with the
distribution $(\ref{shiftedGaussion})$ given by the BD model 
(black solid line).}
(b) {\it The time evolution of the mean square displacement computed by
$10^5$ realizations of the three-dimensional multiscale 
algorithm (red solid line) is compared with the limiting
BD model $(\ref{BDXeq})$--$(\ref{BDVeq})$ (black dashed line) and 
$\sqrt{6 D \, t}$ (blue dot-dashed line). \hfill\break
We use $\mu = 10^3,$ $\gamma = 10$, $D=1$, $\Delta t = 10^{-6}$,
$L=5$, $R=1$, ${\mathbf X}(0) =  [0,0,0]$ and 
${\mathbf V}(0) = [0,0,0]$.
}}
\label{figure4}
\end{figure}
(\ref{shiftedGaussion}). In Figure \ref{figure4}(b), we plot the 
time evolution of the mean square displacement.  The mean 
square displacement corresponding to the limiting
BD model (\ref{BDXeq})--(\ref{BDVeq}) is given by
$(\ref{eqmsd15})$ multiplied by $\sqrt{3}$ because we have
three spatial dimensions (black dashed line). As expected, 
the models compare well. The mean square displacement 
obtained for the BD model (\ref{BDSDE}) is plotted 
as the blue dot-dashed line. 

\section{Application to protein binding to receptors}
\label{bioapplication}

\begin{figure}[t]
\centerline{
\epsfig{file=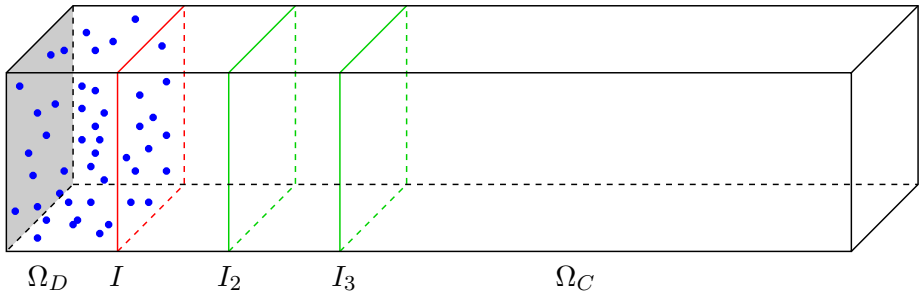,height=3.3cm}
}
\hskip 10mm $\Omega_D$ \hskip 4mm $I$  \hskip 11mm $I_2$ 
\hskip 10.5mm $I_3$ \hskip 25mm $\Omega_C$
\caption{{\it Schematic of the computational domain used in
the protein binding example.}}
\label{figureprotbind}
\end{figure}

In this section, we apply our results to a simplified model of protein
binding to receptors on the cell membrane. We consider simple
geometry which is schematically shown in Figure \ref{figureprotbind}.
Our computational domain is a part of the intracellular space
next to the cell membrane given as the cuboid
$\Omega = [0,L_1] \times [0,L_2] \times [0,L_2]$ 
where $L_1>0$ and $L_2>0$. The cell membrane is modelled
by one side of the cuboid, namely 
$$
\partial \Omega_M = \{0\} \times [0,L_2] \times [0,L_2],
$$
which is shaded gray in Figure \ref{figureprotbind}.
Our goal is to model the binding of diffusing proteins to receptors
on the cell membrane with an MD-level of detail. Therefore we
define $\Omega_D$ as a part of the intracellular space which is
close to the cell membrane $\partial \Omega_M$, i.e.
$$
\Omega_D = [0,h] \times [0,L_2] \times [0,L_2],
\qquad
\mbox{and}
\qquad
\Omega_C = [h,L_1] \times [0,L_2] \times [0,L_2],
$$
where $h>0$ and the interface $I$ is at $x_1 = h$. Diffusing proteins
are modelled as spheres of radius $R$. We consider that 
a protein which hits the boundary $\partial \Omega_M$ will bind to a receptor
with probability $P$, and otherwise it is reflected. This type
of a reactive boundary condition is common for BD simulations 
\citep{Erban:2007:RBC}. In the case of MD, more detailed models
of protein binding could be introduced in $\Omega_D$ 
\citep{Dror:2011:PMD,Vilaseca:2013:UMC}. 
However, the main goal of this section is to show how an MD
model in $\Omega_D$ can be coupled with BD simulators 
which have been developed for simulations of intracellular processes.
Therefore we keep the MD model in $\Omega_D$ as simple as possible.

If we used BD model (\ref{BDXeq})--(\ref{BDVeq}) in $\Omega_C$,
then the situation would be more or less the same as in 
Section \ref{secfromMDtoBD3D}. However, modern BD simulators
of intracellular processes work with the high-friction limit
(\ref{BDSDE}) rather than (\ref{BDXeq})--(\ref{BDVeq}).
For example, the software package Smoldyn discretizes
(\ref{BDSDE}) with a fixed time step and uses (\ref{eq4}) to 
update positions of diffusing proteins. In particular, it uses 
larger values of time step than we used in Section \ref{secfromMDtoBD3D}. 
Then the problem can be formulated as follows: we would like to use 
the MD model with time step $\Delta t$ in $\Omega_D$ and couple it
with the BD model (\ref{eq4}) with larger time step $\overline{\Delta t}$,
namely
\begin{equation}
X_i(t+\overline{\Delta t}) = X_i(t) + \sqrt{2D \overline{\Delta t}} \, \xi_{i},
\qquad
i=1,2,3,
\label{eq4overline}
\end{equation}
if the diffusing molecule is far away from $\Omega_D$. We 
couple these models using the intermediate BD model 
(\ref{BDXeq})--(\ref{BDVeq}). We introduce two 
additional interfaces 
$$
I_2 = \{h_2\} \times [0,L_2] \times [0,L_2],
\quad
\mbox{and}
\quad
I_3 = \{h_3\} \times [0,L_2] \times [0,L_2],
$$
where $h < h_2 < h_3 < L_1$, as shown in Figure \ref{figureprotbind}. 
We denote
$$
\Omega_{C1}
=
[h,h_3] \times [0,L_2] \times [0,L_2],
\quad
\mbox{and}
\quad
\Omega_{C2} 
= 
[h_2,L_1] \times [0,L_2] \times [0,L_2],
$$
i.e. $\Omega_{C1}$ and $\Omega_{C2}$ are two overlapping subdomains
of $\Omega_C$. We simulate the time evolution of the position
$\bold{X}(t)$ of one protein molecule. If $\bold{X}(t) \in \Omega_{C2}$,
then the BD model (\ref{eq4overline}) will be used in $\Omega_{C2}$ until
the molecule leaves $\Omega_{C2}$. Then we switch to the shorter
time step $\Delta t$ and use the BD model (\ref{BDXeq})--(\ref{BDVeq})
in $\Omega_{C1}$. The protein molecule can leave 
$\Omega_{C1}$ in two possible ways:

\medskip

{

\leftskip 15mm

\parindent -5.4mm

(i) The protein molecule crosses the interface $I_3.$ \\
Then we revert to the BD model (\ref{eq4overline}) which 
is used in $\Omega_{C2}$.

\medskip

\parindent -6.6mm

(ii) The protein molecule crosses the interface $I.$ \\
Then we use the method from Section \ref{secfromMDtoBD3D}
for coupling the MD model in $\Omega_D$ with 
the BD model (\ref{BDXeq})--(\ref{BDVeq}) in $\Omega_{C1}$.

\leftskip 0mm

}

\medskip

\noindent
Since the subdomains $\Omega_{C1}$ and $\Omega_{C2}$ overlap,
we can use the limiting result (\ref{shiftedGaussion}) which 
implies that, for times $t \ge \gamma^{-1}$, the BD model 
(\ref{BDXeq})--(\ref{BDVeq}) is given by the BD model 
(\ref{eq4overline}) shifted by time $t^{*} = 3/(2\gamma)$.
In particular, we will also add or subtract $t^{*}$ from the
time variable whenever we switch between BD models.

In Figure \ref{figure6}, we present illustrative results computed
by averaging over $10^5$ realizations. Initial positions
of the protein molecule are uniformly distributed along the $x_1$
axis. The histogram of positions (along the $x_1$-axis) at time 
$t=1$ is plotted in Figure \ref{figure6}(a). 
\begin{figure}[t]
\picturesABal{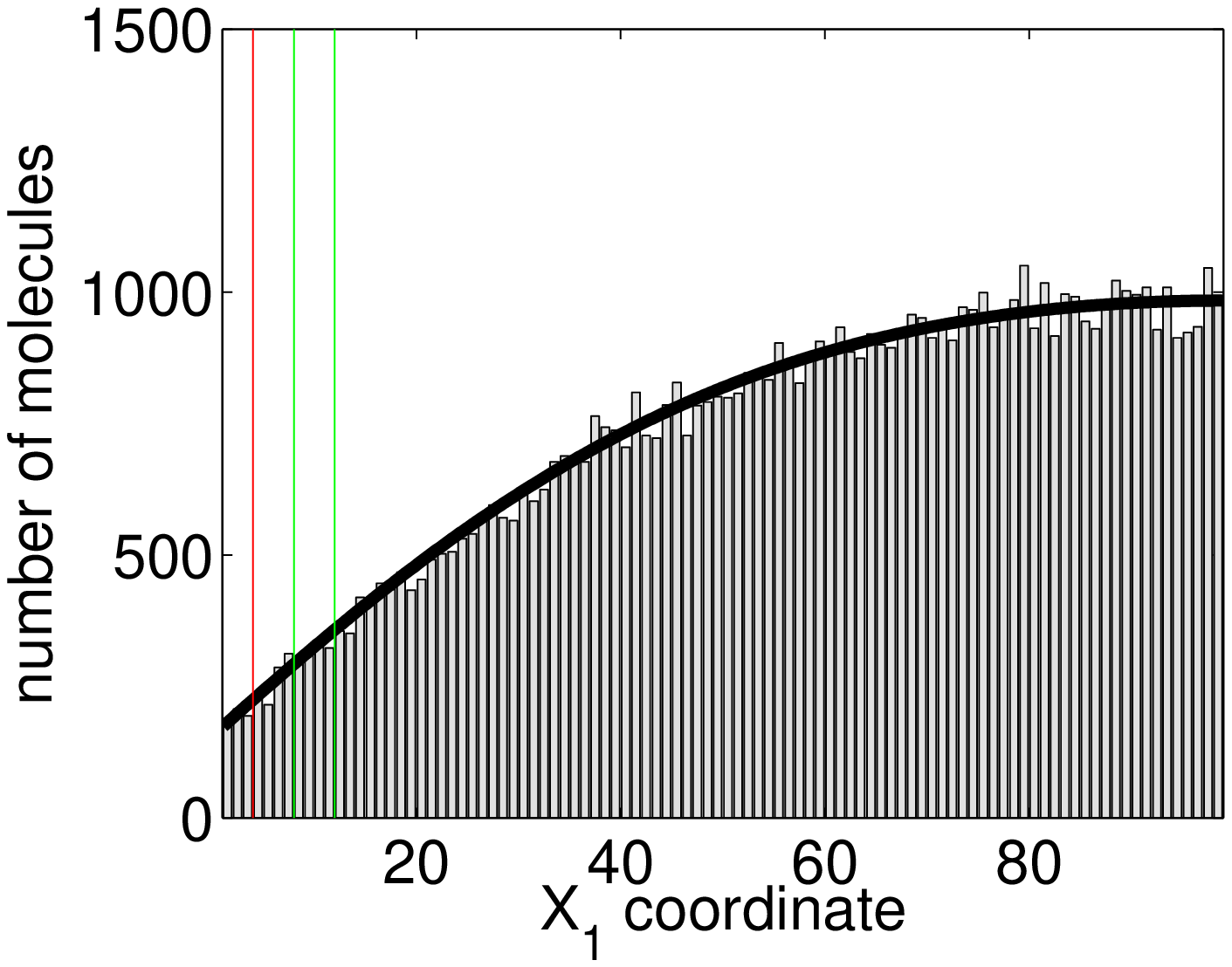}{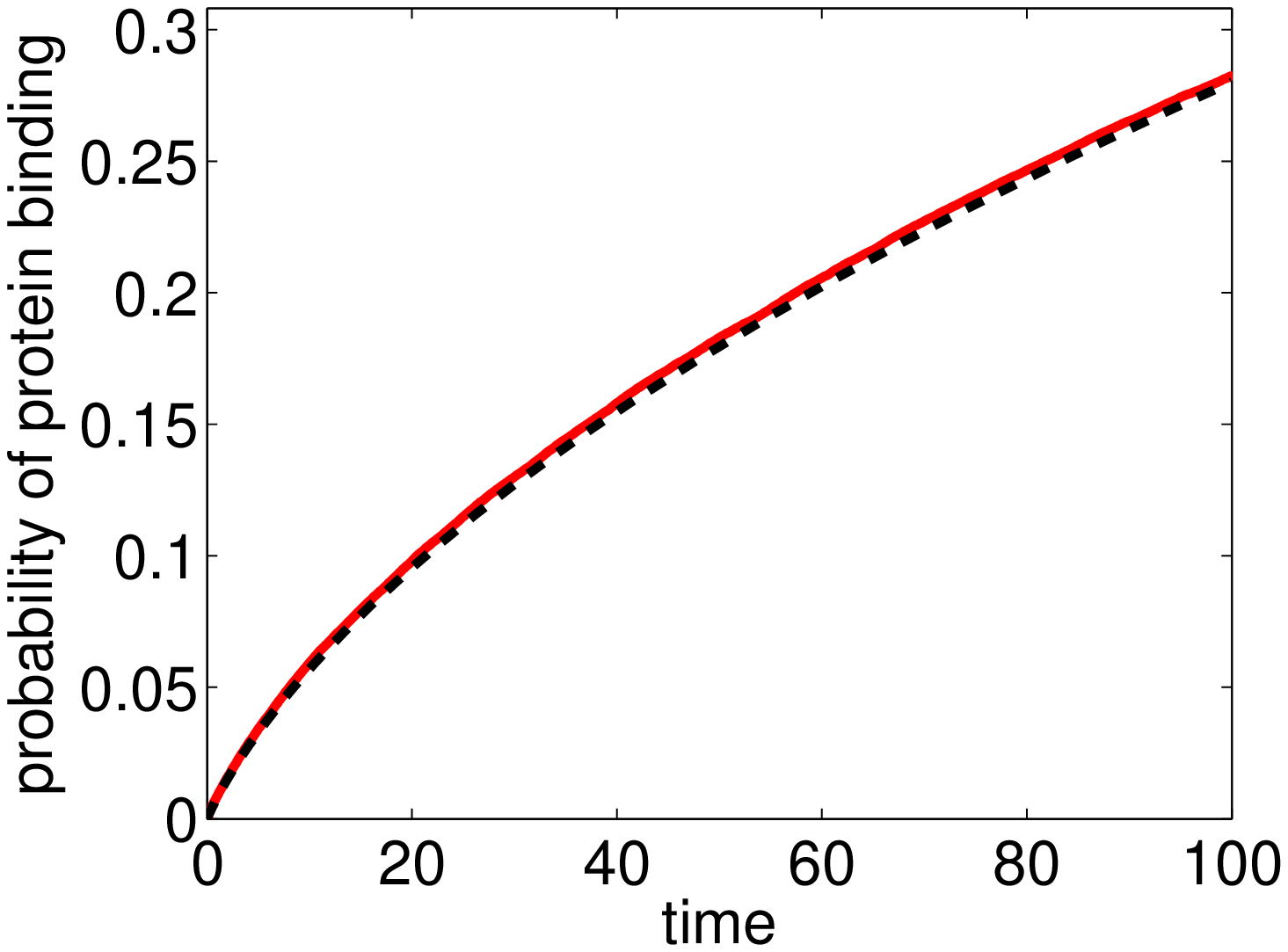}{5.3cm}{3mm}{5.1cm}
\caption{(a) {\it Distribution of positions along the $x_1$-axis
at time $t=1$ computed by the multiscale model described in 
Section $\ref{bioapplication}$ (gray histogram). 
The black solid line is the solution of the limiting PDE model 
$(\ref{difprotbind1})$--$(\ref{difprotbind2})$. Vertical lines
denote interfaces $I$, $I_2$ and $I_3$.}
(b) {\it Probability that the protein is bind to a receptors 
as a function of time $t$ computed by the multiscale
model (red solid line) and the PDE model 
$(\ref{difprotbind1})$--$(\ref{difprotbind2})$
(black dashed line). Parameters used:
$D=10$, $\gamma = 10^2$, $\mu=10^3$, $K=1$, 
$R=1$, $L_1 = 10^2$, $h = 4$, $h_2 = 8$, $h_3=12$
and $h_4 = 3$. 
}}
\label{figure6}
\end{figure}
Interfaces $I$, $I_2$ and $I_3$ are also shown in this plot. Since
we used a very simple model of the protein binding, we can compare it
with the mean-field limit given by the solution of the partial
differential equation (PDE)
\begin{equation}
\frac{\partial \varrho}{\partial t}(x_1,t)
=
D \,
\frac{\partial^2 \varrho}{\partial x_1^2}
(x_1,t),
\qquad
x_1 \in [0,L_1], 
\quad
t \ge 0,
\label{difprotbind1}
\end{equation}
with boundary conditions \citep{Erban:2007:RBC}
\begin{equation}
D \, \frac{\partial \varrho}{\partial x_1}
(0,t)
=
K \, \varrho (0,t),
\qquad
D \frac{\partial \varrho}{\partial x_1}
(L_1,t)
=
0,
\qquad
\mbox{where}
\qquad
P = \frac{K \sqrt{2 \pi}}{\sqrt{D \gamma}}.
\label{difprotbind2}
\end{equation}
The solution of (\ref{difprotbind1})--(\ref{difprotbind2}) with
uniform initial condition $\varrho(x_1,0) \equiv \mbox{const}$
is given by the black solid line in Figure \ref{figure6}(a). 
Since we only visualize the distribution along the $x_1$ axis
in Figure \ref{figure6}(a), we can further decrease the 
computational cost by truncating the simulation domain in the 
$x_2$ and $x_3$ directions to the region close to the protein 
molecule. That is, we only simulate small particles in 
the subdomain $[0,h] \times [X_2(t') - h_4, X_2(t') + h_4] \times 
[X_3(t') - h_4, X_3(t') + h_4]$ where $t'$ is the time when 
the protein molecule enters $\Omega_D \cup \Omega_{C1}$.
This subdomain  (moving window) is shifted accordingly whenever $X_2(t)$ 
or $X_3(t)$ approach its boundary.

The probability that the protein is adsorbed to the surface
is given as a function of time in Figure \ref{figure6}(b). 
It again compares well with the results obtained by the 
limiting PDE system (\ref{difprotbind1})--(\ref{difprotbind2}).

\section{Discussion}

\label{secdiscussion}

I have presented and analysed a multiscale approach which uses 
MD simulations in a part of the computational domain and BD simulations 
in the rest of the domain. The ultimate goal of this research is
to use MD to help parameterize BD models of intracellular processes.
One application area is modelling proteins in an aquatic environment
which is useful for understanding protein binding to receptors 
(surfaces) as shown in Section \ref{bioapplication}. 

The main idea of the presented coupling of MD and BD models is based
on using equations (\ref{halfSDE}) and (\ref{halfSDE3D}) and
estimating drift and diffusion coefficients for velocities of
molecules which cross the interface $I$. This coupling uses the
same time step for the BD model (\ref{BDXeq})--(\ref{BDVeq})
as for the MD model. In Section \ref{bioapplication}, it was shown
that this is not a limiting step of this approach, because the BD
model (\ref{BDXeq})--(\ref{BDVeq}) is only needed in a small part
of the domain next to $\Omega_D$. Then the coarser BD model 
(\ref{eq4overline}) with larger time step can be used in the rest 
of the simulation domain, using a suitable overlap region. Another 
overlap region could be used to couple BD simulations with 
mean-field PDE-based models \citep{Franz:2013:MRA}. Then multiscale
models which couple BD (of point particles) with coarser 
reaction-diffusion approaches would be capable of further
increasing time scales and space scales of simulations
\citep{Flegg:2012:TRM,Franz:2013:MRA}.

MD models considered in this paper are relatively simple and analytically 
tractable, describing water molecules as point particles. An important 
generalization is to consider more complicated MD models of 
water molecules \citep{Huggins:2012:CLW}. For example, 
\cite{Rahman:1971:MDS} model water molecules 
as rigid asymmetric rotors. That is, each water molecule is 
described by six coordinates: the position of its centre of 
mass and three angles describing molecule orientation. 
The energy of water solution is given as the sum of kinetic 
energies (for translation and rotation) and the intermolecular 
potential which is assumed to be pairwise additive and can 
be given in several different ways, i.e. the heat bath is given
by its Hamiltonian \citep{Rahman:1971:MDS,Huggins:2012:CLW,Mark:2001:SDT}.
I am currently investigating MD models based on Hamiltonian 
dynamics, with the aim of designing and analyzing multiscale
algorithms similar to the algorithm [M1]--[M8] from this 
paper. The ultimate goal of this research is to design BD models
of intracellular process which make use of modern MD simulations 
\citep{Merz:2010:LFE,Deng:2009:CSB} to infer parameters of BD 
models \citep{Lipkova:2011:ABD}.
I will report my results in a future publication.

\section*{Acknowledgements}

I would like to thank the Royal Society for a University Research 
Fellowship; Brasenose College, University of Oxford, for a Nicholas 
Kurti Junior Fellowship; and the Leverhulme Trust for a Philip 
Leverhulme Prize. The research leading to these results has received 
funding from the European Research Council under the {\it European 
Community's} Seventh Framework Programme {\it (FP7/2007-2013)} /
ERC {\it grant agreement} No. 239870.

\end{document}